\documentclass[12pt, draftclsnofoot, onecolumn]{IEEEtran}
\IEEEoverridecommandlockouts
\usepackage{hyperref}
\usepackage{cite}
\usepackage{amsmath,amssymb,amsfonts,amsthm}
\usepackage{graphicx}
\usepackage{textcomp}
\usepackage[dvipsnames]{xcolor}
\usepackage{bm}
\usepackage{arydshln}
\usepackage{multirow}
\usepackage{mathdots}
\usepackage{mathtools}
\usepackage{setspace}
\usepackage{pgfplots}
\usepackage{tikz,pgfplots,tikzpeople}
\usetikzlibrary{positioning} 
\usepackage{nicematrix}
\usetikzlibrary{patterns}
\usepackage[labelformat=simple]{subcaption}
\usepackage{wasysym}

\newtheorem{theorem}{Theorem}

\newtheorem{lemma}{Lemma}

\newtheorem{proposition}{Proposition}

\let\emptyset\varnothing

\DeclarePairedDelimiter\ceil{\lceil}{\rceil}

\usepackage{algpseudocode,float}
\usepackage[ruled,linesnumbered]{algorithm2e}

\usepackage{lipsum}

\allowdisplaybreaks

\def\BibTeX{{\rm B\kern-.05em{\sc i\kern-.025em b}\kern-.08em
    T\kern-.1667em\lower.7ex\hbox{E}\kern-.125emX}}
\title{How to Read and Update Coded Distributed Storage Robustly and Optimally?}
\author{
    Haobo~Jia
    and Zhuqing~Jia,~\IEEEmembership{Member,~IEEE}
    \thanks{H. Jia and Z. Jia are with the School of Artificial Intelligence, Beijing University of Posts and Telecommunications, Beijing, 100086 China (e-mail: jiahaobo@bupt.edu.cn; zhuqingj@bupt.edu.cn).}
}

\begin{document}
\maketitle
\begin{abstract}
    We consider the question of how to robustly and optimally read and update coded distributed storage, namely the problem of robust dynamic coded distributed storage (RDCDS) that is associated with the $(K_c,R_r,N)$-coded distributed storage of a message with $N$ servers where 1) it suffices to recover the message from the storage at any $R_r$ servers; and 2) each of the servers stores a coded portion of the message that is at most $\frac{1}{K_c}$ the size of the message. The goal is to enable two main functionalities: the read operation and the update operation of the message. Specifically, at time slot $t$, the user may execute either the read operation or the update operation, where the read operation allows the user to recover the message from the servers by downloading symbols, and the update operation allows the user to update the message to the servers in the form of an additive increment by uploading $X^{(t)}$-securely coded symbols so that any up to $X^{(t)}$ colluding servers reveal nothing about the increment. The two functionalities are robust if 1) they tolerate temporarily dropout servers up to certain thresholds at any time slot $t$, i.e., the read/update operation remains feasible at any time slot $t$ as long as the number of available servers exceeds a certain threshold (the read threshold is $R_r$ by definition and the update threshold is denoted as $R_u^{(t)}$); and 2) the user may remain oblivious to prior server states at any time slot $t$, i.e., no history information about server states is required during read/update operation. The communication efficiency of the two functionalities at time slot $t$ is measured by the download cost $C_r^{(t)}$ (i.e., the number of message symbols recovered per downloaded symbol) of the read operation and the upload cost $C_u^{(t)}$ (i.e., the number of message symbols updated per uploaded symbol) of the update operation. Given the storage cost factor $K_c$ and the read threshold $R_r$ where $0<K_c\leq R_r\leq N$, we are curious about the following question: what is the optimal $(R_u^{(t)},C_r^{(t)},C_u^{(t)})$ tuple? In this work, we answer the question and thus settle the fundamental limits of RDCDS. In particular, denoting the number of dropout servers at time slot $t$ as $|\mathcal{D}^{(t)}|$, we first show that 1) $R_u^{(t)}\geq N-R_r+\ceil*{K_c}+X^{(t)}$; and 2) $C_r^{(t)}\geq \frac{N-|\mathcal{D}^{(t)}|}{N-R_r+\ceil*{K_c}-|\mathcal{D}^{(t)}|}, C_u^{(t)}\geq \frac{N-|\mathcal{D}^{(t)}|}{R_r-X^{(t)}-|\mathcal{D}^{(t)}|}$. Then, inspired by the idea of staircase codes, we construct an RDCDS scheme that simultaneously achieves the above lower bounds on $R_u^{(t)},C_r^{(t)}$ and $C_u^{(t)}$. The technical aspects of our achievability scheme build upon the following novelties: 1) a novel staircase structure that minimizes the upload cost; 2) a nullspace design that optimally tolerates dropout servers during update operations; and 3) a memoryless update mechanism that requires no knowledge on prior server states or additional transactions/server-to-server communications to maintain storage consistency.

\end{abstract}

\section{Introduction}

Coded distributed storage refers to a distributed storage system (DSS) in which each server stores a codeword of the message according to a selected storage code so that the DSS is 1) efficient in terms of storage and communication costs, 2) resilient against server dropouts and failures\cite{Dimakis2006,Blaum1995,weatherspoon2002erasure,Wang2018}, and even 3) capable of efficient failure repair
\cite{dimakis2010network,Dimakis2011,Rashmi2011,Cadambe2013,Pawar2011,gopalan2012locality,Tamo2014,Cadambe2015}. 
In this work, we focus on the $(K_c,R_r,N)$-coded distributed storage, which consists of $N$ distributed servers. In this setup, a user can recover the message from any $R_r$ servers, while the storage cost at each server is limited to $\frac{1}{K_c}$ of the message size, where $0<K_c\leq N$. We address the fundamental challenge of enabling simultaneously read and update functionalities in the $(K_c,R_r,N)$-coded distributed storage. Specifically, our goal is to construct a $(K_c,R_r,N)$-coded distributed storage and associated schemes that allow a user, associated with time slot $t$, to either read/recover the message from the downloaded symbols or update the message with an additive increment by uploading $X^{(t)}$-securely coded increment symbols that disclose no information about the increment to any up to $X^{(t)}$ colluding servers (hence the {\it dynamic} aspect of the coded distributed storage accounts for ongoing message updates). Given the inherent uncertainties in DSS, we assume that at any time slot $t$, there may be a set of ``dropout'' servers that are temporarily unavailable to fulfill the user's request -- whether for reading or updating. The set of dropout servers is randomly determined at each time slot (so dropout servers at time slot $t$ may become available again in the future) but, once established, is considered globally known (to both the user and the available servers) and remains unchanged during that time slot. Since keeping track of past server states adds storage and communication costs, it is preferable for the user to remain unaware of prior server states. Incorporating robustness considerations, the problem of {\it robust} dynamic coded distributed storage (RDCDS) requires that the read and update schemes remain feasible as long as the number of available servers exceeds certain thresholds (referred to as the read threshold $R_r$ and the update threshold $R_u^{(t)}$, respectively), with no need for historical server state information. The objective of RDCDS is to enable the most robust and efficient read and update functionalities, allowing the scheme to fully exploit and adapt to the available servers to minimize communication costs, including the download cost for the read operation $C_r^{(t)}$ and the upload cost for the update operation $C_u^{(t)}$.

It is important to momentarily distinguish between updating a coded distributed storage by re-encoding the updated message and using an additive increment update. When a user has the updated message, there is a quite straightforward scheme to update the DSS: the user simply re-encodes the message and uploads it to the available servers. However, this work explores a scenario where the user wishes to update the message through an additive increment. Specifically, let the current message be denoted as \(\mathbf{W} \in \mathbb{F}_q^{L}\) and the increment as \(\boldsymbol{\Delta} \in \mathbb{F}_q^{L}\). The updated message is then defined as the sum \(\mathbf{W} + \boldsymbol{\Delta}\). Moreover, we consider the case where the user generates coded increments without prior knowledge of the current message \(\mathbf{W}\) or the server storage. Therefore, our results, including both converse bounds and achievability schemes, are applicable in scenarios where the user can generate the increment independently of the current message, thereby eliminating the cost to read or recover it. Additionally, these results are relevant when the user seeks to update the DSS via additive increment, where the procedure for generating coded increments depends solely on the increment and optional noise symbols used to protect it. There are numerous scenarios where independent increments can be applied without knowledge of the previous value. For instance, in financial systems, a series of independent transactions can increment a counter (such as an account balance) by varying amounts, with each increment applied without needing to know the prior balance. Similarly, in data aggregation within sensor networks, each increment represents new sensor data, and the message corresponds to the aggregated value of all collected data. In federated learning, updates to the global machine learning model are computed based on local data at each client. These updates (gradients) are aggregated across a set of asynchronous clients, where within each batch, each client can update the global model independently, without requiring the model state from previous clients. Moreover, generating coded increments without needing access to the current or updated message may provide additional security benefits, as, in the worst case, the update process reveals nothing beyond the increment itself.

The main result of this work is a comprehensive resolution of the following question: Given the storage cost factor $K_c$ and the read threshold $R_r$, what is the optimal $(R_u^{(t)}, C_r^{(t)}, C_u^{(t)})$ tuple? Specifically, denoting the number of dropout servers at time slot $t$ as $|\mathcal{D}^{(t)}|$, we establish that the following three lower bounds hold: $R_u^{(t)} \geq N - R_r + \lceil K_c \rceil + X^{(t)}$, $C_r^{(t)} \geq \frac{N - |\mathcal{D}^{(t)}|}{N - R_r + \lceil K_c \rceil - |\mathcal{D}^{(t)}|}$, and $C_u^{(t)} \geq \frac{N - |\mathcal{D}^{(t)}|}{R_r - X^{(t)} - |\mathcal{D}^{(t)}|}$. Then, inspired by the concept of staircase codes \cite{bitar2017staircase}, we construct an RDCDS scheme that achieves these lower bounds. An intuitive explanation of these optimality results is provided in Section \ref{sec:competing}, and the intuition behind the converse proof is outlined in Section \ref{sec:cvrseintui}.

While several related works exist, such as the adaptive cross-subspace alignment read and write (ACSA-RW) scheme \cite{Jia_Jafar_XSTPFSL} for {\it private} read and update, and the staircase code \cite{bitar2017staircase} for communication-efficient read in securely-coded distributed storage, our contribution is novel in two key aspects. First, although the ACSA-RW scheme supports private read and update, ensuring that the user's query is kept private from any up to $T$ servers, and our problem can be viewed as a special case of private read and update with $T=0$, however, no complete converse results have been established for private read and update thus far. Therefore, the optimality results in this work represent a significant step towards the settlement of the private read and update problem. Moreover, our results strictly improve upon the ACSA-RW scheme when applied to the $T=0$ case. Further discussion on this can be found in Section \ref{sec:competing}. Second, although our achievability scheme draws inspiration from the staircase code \cite{bitar2017staircase}, it turns out that a novel staircase structure is necessary to optimize the upload cost. Details of this construction are discussed in Section \ref{sec:comparesc}. Additionally, the mechanism of nullspace design for tolerating dropout servers, specifically the ``ACSA Null-shaper'' building block in \cite{Jia_Jafar_XSTPFSL}, is generalized for the staircase structure. We refer the readers to Section \ref{sec:compareacsa} for more information on this generalization.

The remainder of this paper is organized as follows. In Section \ref{sec:ps}, we formally define the problem of RDCDS, and Section \ref{sec:mainres} presents the main results of this work as Theorem \ref{thm:main}, along with several key observations. Section \ref{sec:converse} is dedicated to the converse proof of Theorem \ref{thm:main}, and in Section \ref{sec:achiv}, we present our achievability scheme for the theorem. Finally, we conclude the paper in Section \ref{sec:conclu}. 

\emph{Notation:} Bold symbols are used to denote vectors and matrices, while calligraphic symbols denote sets. Following the convention, let the empty product be the multiplicative identity, and the empty sum be the additive identity. For two positive integers $M,N$ such that $M\le N,[M:N]$ denotes the set $\{M,M+1,\cdots,N\}$. We use the shorthand notation $[N]$ for $[1:N]$. $\mathbb{N}$ denotes the set of non-negative integers $\{0,1,2,\cdots\}$, and $\mathbb{N}^*$ denotes the set $\mathbb{N}\setminus\{0\}$. For a subset of integers $\mathcal{C}$, $\mathcal{C}(i),i\in[|\mathcal{C}|]$ denotes its $i^{th}$ element in ascending order. For a row (column) vector $\mathbf{v}$ of dimension $n$, $\mathbf{v}(\mathcal{I}), \mathcal{I}\subset[n]$ denotes the row (column) vector formed by the entries indexed by $\mathcal{I}$. We use the shorthand notation $\mathbf{v}(i)$ for $\mathbf{v}(\{i\})$, i.e., the $i^{th}$ entry of $\mathbf{v}$. For an $m\times n$ matrix $\mathbf{V}$ and two sets $\mathcal{A}\subset[m],\mathcal{B}\subset[n]$, $\mathbf{V}(\mathcal{A},\mathcal{B})$ denotes the submatrix of $\mathbf{V}$ formed by selecting rows indexed by $\mathcal{A}$ and columns indexed by $\mathcal{B}$. If $\mathcal{A}=[m]$ (or $\mathcal{B}=[n]$), it is abbreviated as the colon operator ($:$) in this context. We use the shorthand notation $\mathbf{V}(a,\mathcal{B}), \mathbf{V}(\mathcal{A},b), \mathbf{V}(a,b)$ for $\mathbf{V}(\{a\},\mathcal{B}),\mathbf{V}(\mathcal{A},\{b\}), \mathbf{V}(\{a\},\{b\})$ respectively. $\mathbf{0}_{m\times n}$ denotes the zero matrix of size $m\times n$.

\section{Problem Statement: RDCDS}\label{sec:ps}

\begin{figure}[!htbp]
    \centering
    \begin{subfigure}{1\columnwidth}
        \centering
        \begin{tikzpicture}[xscale=0.7,yscale=1]

        \node [draw, rectangle,fill=teal!10, text=black, inner sep =0.2cm, rounded corners=0.5em] (S1) at (-7.5cm, 0cm) {\small\begin{tabular}{c}Server $1$ \\ $\mathbf{S}^{(t)}_{1}$\end{tabular}};

        \node [draw, rectangle,fill=teal!10, text=black, inner sep =0.2cm, rounded corners=0.5em] (S2) at (-2.5cm, 0cm) {\small\begin{tabular}{c} Server $2$\\ $\mathbf{S}^{(t)}_{2}$\end{tabular}};

        \node [rectangle, inner sep =0.2cm, rounded corners=0.5em] (Ddots1) at (0.5cm, 0cm) {$\cdots$};

        \node [draw, rectangle,fill=teal!10, text=black, inner sep =0.2cm, rounded corners=0.5em] (S3) at (3.5cm, 0cm) {\small\begin{tabular}{c}Server $i$ \\ $\mathbf{S}^{(t)}_{i}$\end{tabular}};

        \node [rectangle, inner sep =0.2cm, rounded corners=0.5em] (Ddots2) at (6.5cm, 0cm) {$\cdots$};

        \node [draw, rectangle,fill=teal!10, text=black, inner sep =0.2cm, rounded corners=0.5em] (S4) at (9.5cm, 0cm) {\small\begin{tabular}{c}Server $N$ \\ $\mathbf{S}^{(t)}_{N}$\end{tabular}};

        \node[bob,minimum size=1cm] (User) at (1cm, -4cm) {User};

        \draw [black, thick, ->] (S1.south)to node [left=0.5cm] {\small $ \mathbf{A}_{1}^{(t)}$} (User);
        \draw [black, thick, ->] (S2.south)to node [left=0.2cm] {\small $ \mathbf{A}_{2}^{(t)}$} (User);
        \draw [red, dashed, ->] (S3.south)to node [left=0.1cm] {\large $\varnothing$} (User);
        \draw [black, thick, ->] (S4.south)to node [left=0.5cm] {\small $ \mathbf{A}_{N}^{(t)}$} (User);

        \node[below=0.2cm of User] (bk){};
        
        \node[minimum size=0.3cm, inner sep=0.1cm] (W) at (1cm, -6cm) {$\mathbf{W}^{(t)}$};
        \draw [black, thick, ->] (bk)--(W);
    
    \end{tikzpicture}
    \caption{The read operation is executed at time slot $t$, where Server $i$ is unavailable, i.e., $i\in\mathcal{D}^{(t)}$.}
    \vspace{0.5cm}
    \end{subfigure}
    
    \begin{subfigure}{1\columnwidth}
        \centering
        \begin{tikzpicture}[xscale=0.7,yscale=1]

        \node [draw, rectangle,fill=teal!10, text=black, inner sep =0.2cm, rounded corners=0.5em] (S1) at (-7.5cm, 0cm) {\small\begin{tabular}{c}Server $1$ \\ $\mathbf{S}^{(t+1)}_{1}$\\$\uparrow$\\$\mathbf{S}^{(t)}_{1}$\end{tabular}};

        \node [draw, rectangle,fill=teal!10, text=black, inner sep =0.2cm, rounded corners=0.5em] (S2) at (-2.5cm, 0cm) {\small\begin{tabular}{c} Server $2$\\ $\mathbf{S}^{(t+1)}_{2}$\\$\uparrow$\\ $\mathbf{S}^{(t)}_{2}$\end{tabular}};

        \node [rectangle, inner sep =0.2cm, rounded corners=0.5em] (Ddots1) at (0.5cm, 0cm) {$\cdots$};

        \node [draw, rectangle,fill=teal!10, text=black, inner sep =0.2cm, rounded corners=0.5em] (S3) at (3.5cm, 0cm) {\small\begin{tabular}{c}Server $i$\\ $\mathbf{S}^{(t+1)}_{i}$\\$\parallel$ \\ $\mathbf{S}^{(t)}_{i}$\end{tabular}};

        \node [rectangle, inner sep =0.2cm, rounded corners=0.5em] (Ddots2) at (6.5cm, 0cm) {$\cdots$};

        \node [draw, rectangle,fill=teal!10, text=black, inner sep =0.2cm, rounded corners=0.5em] (S4) at (9.5cm, 0cm) {\small\begin{tabular}{c}Server $N$\\ $\mathbf{S}^{(t+1)}_{N}$\\$\uparrow$ \\ $\mathbf{S}^{(t)}_{N}$\end{tabular}};

        \node[bob,minimum size=1cm] (User) at (1cm, -4cm) {User};
        
        \draw [black, thick,  ->] (User)to node [right=0.5cm] {\small $ \mathbf{Q}_{1}^{(t)}$}(S1.south);
        \draw [black, thick,  ->] (User)to node [right=0.2cm] {\small $ \mathbf{Q}_{2}^{(t)}$}(S2.south); %
        \draw [red, dashed,  ->] (User)to node [right=0.1cm] {\large $\varnothing$}(S3.south);
        \draw [black, thick,  ->] (User)to node [right=0.5cm] {\small $ \mathbf{Q}_{N}^{(t)}$}(S4.south);
        
        \node[below=0.2cm of User] (bk){};
        
        \node[minimum size=0.3cm, inner sep=0.1cm] (W) at (1cm, -6cm) {$\boldsymbol{\Delta}^{(t)}$};
        \draw [black, thick, ->] (W)--(bk);
        
        \end{tikzpicture}
        \caption{The update operation is executed at time slot $t$, where Server $i$ is unavailable, i.e, $i\in\mathcal{D}^{(t)}$. Note that the storage at Server $i$ cannot be updated, $\mathbf{S}_i^{(t+1)}=\mathbf{S}_i^{(t)}$.}

    \end{subfigure}
    
    \caption{The problem of robust dynamic coded distributed storage (RDCDS).}
    \label{fig:RDCDS}
\end{figure}
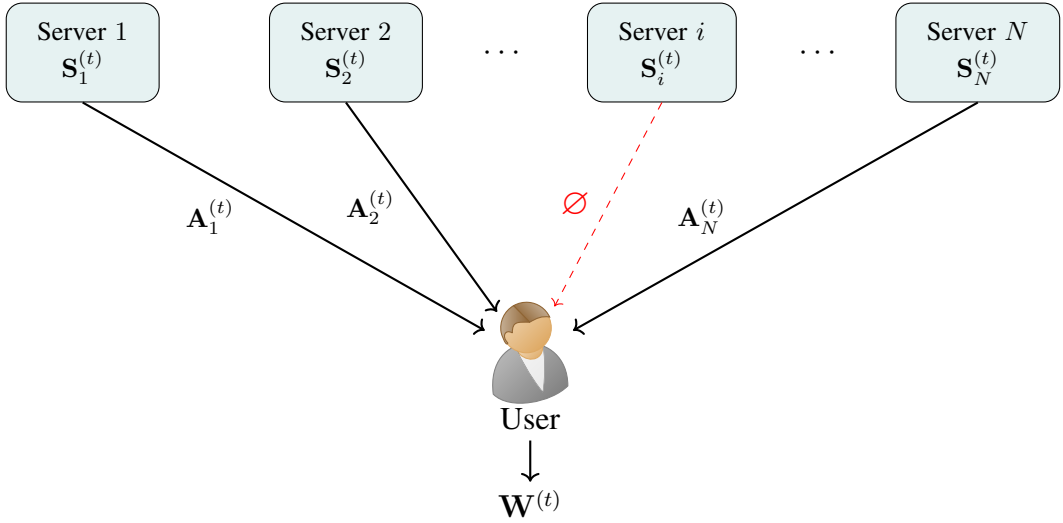
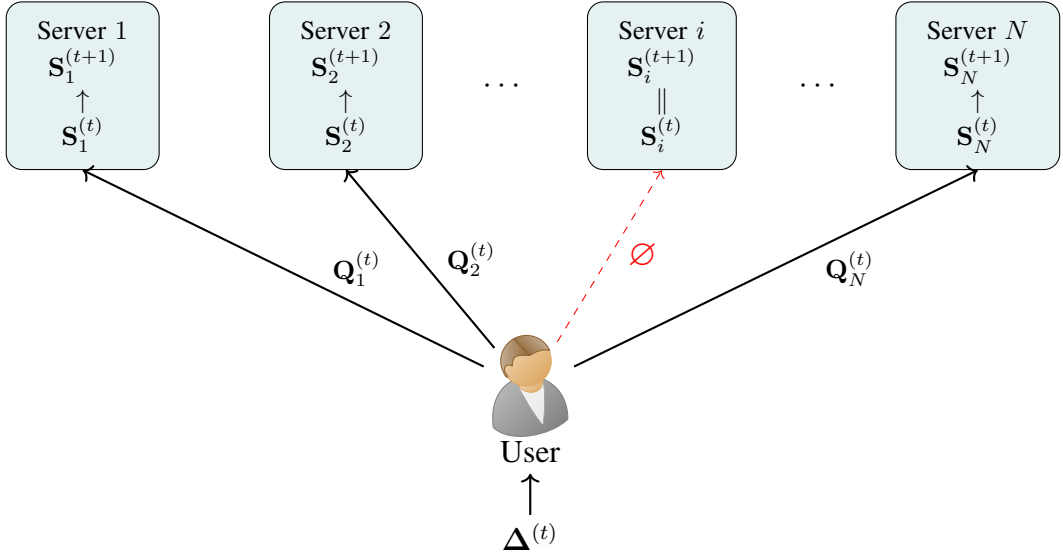

Consider a distributed storage system (DSS) of $N$ servers. As shown in Figure \ref{fig:RDCDS}, the problem of RDCDS is associated with the coded distributed storage of a message over time slots $t$, i.e., for all $t\in\mathbb{N}$, the collection of the storage at all servers must represent an $(K_c,R_r,N)$-secure storage of the message $\mathbf{W}^{(t)}$ that consists of $L$ (i.i.d.) symbols from the finite field $\mathbb{F}_q$, where $0<K_c\leq R_r\leq N, K_c\in\mathbb{R}$. In other words, denoting the storage at Server $n$ as $\mathbf{S}_n^{(t)}$, for all $t\in\mathbb{N}$, we have
\begin{itemize}
    \item {\bf $R_r$-recoverability: }The message must be a deterministic function of the storage at any $R_r$ servers, i.e., for all $\mathcal{R}\subset [N]$ such that $|\mathcal{R}|=R_r$,
    \begin{align}
        H(\mathbf{W}^{(t)}\mid (\mathbf{S}_r^{(t)})_{r\in\mathcal{R}})=0.\label{def:Rr-rec}
    \end{align}
    \item {\bf $K_c$-storage cost: }The storage at any server is at most $\frac{1}{K_c}$ the size of the message, i.e., for all $n\in[N]$,
    \begin{align}\label{eq:minstor}
        H(\mathbf{S}_n^{(t)})\leq \frac{L}{K_c}
    \end{align}
    in $q$-ary units.
\end{itemize}
Note that the initial secure storage $(\mathbf{S}_n^{(0)})_{n\in[N]}$ and the corresponding message $\mathbf{W}^{(0)}$ is initialized {\it a priori}, e.g., by the global coordinator, etc.

There is a series of users, and for each time slot $t, t\in\mathbb{N}^*$, there is one user associated with it. At any time slot $t$, the user may wish to execute either the read operation or the update operation. Due to various uncertainties in the DSS, a subset of servers may be temporarily unavailable to respond, which are referred to as dropout servers. We assume that the set of dropout servers, denoted as $\mathcal{D}^{(t)}$, is globally known prior to the operation and remains constant during the time slot. Besides, we also assume that the servers experience all possible states for read and update operations within finite time slots, i.e., there exists a positive integer $t_0$ such that at time slots $t_1<t_2<\cdots<t_0$, the user executes a series of read and update operations with all possible constraints and the servers experience all possible dropouts $\mathcal{D}^{(t)}, t=t_1,t_2,\cdots,t_0$ during these operations. Recall that robustness requires that the user must remain oblivious to prior server states (including the current time index $t$ which corresponds to the number of read and update operations that have already been executed), therefore it only makes sense if we focus on the steady state of the system, i.e., for sufficiently large $t>t_0$.

If the read operation is executed at time slot $t$, the user downloads symbols from the available servers $[N]\setminus\mathcal{D}^{(t)}$ to recover the message $\mathbf{W}^{(t)}$. Denoting downloaded symbols from Server $n$ as $\mathbf{A}_n^{(t)}$, we have
\begin{itemize}
    \item {\bf Determinacy: }The downloaded symbols from each of the servers must be a deterministic function of its storage, i.e., for all $n\in[N]\setminus\mathcal{D}^{(t)}$,
    \begin{align}
        H(\mathbf{A}_n^{(t)}\mid \mathbf{S}_n^{(t)})=0.\label{def:detmin}
    \end{align}
    \item {\bf Correctness: }The message must be recoverable from the downloads, i.e.,
    \begin{align}
        H(\mathbf{W}^{(t)}\mid (\mathbf{A}_{n}^{(t)})_{n\in [N]\setminus\mathcal{D}^{(t)}})=0.\label{def:dlcrec}
    \end{align}
    \item {\bf Storage transition: }The storage at each server is untouched by the read operation, i.e., for all $n\in[N]$,
    \begin{align}\label{def:stortranrd}
        \mathbf{S}_n^{(t+1)}=\mathbf{S}_n^{(t)}
    \end{align}
    and
    \begin{align}\label{def:wtranrd}
        \mathbf{W}^{(t+1)}=\mathbf{W}^{(t)}.
    \end{align}
\end{itemize}
For an RDCDS scheme, the communication efficiency of the read operation is characterized by the normalized download cost $C_r^{(t)}$, defined as
\begin{align}
    C_r^{(t)}=\frac{\sum_{n\in [N]\setminus\mathcal{D}^{(t)}} H(\mathbf{A}_n^{(t)})}{L}.
\end{align}

On the other hand, if the update operation is executed at time slot $t$, the message stored in the DSS is updated by the user-generated increment $\boldsymbol{\Delta}^{(t)}$ consisting of $L$ (i.i.d.) symbols from the finite field $\mathbb{F}_q$. To this end, each of the available servers updates its storage according to the $X^{(t)}$-securely coded increment uploaded by the user, denoted as $\mathbf{Q}_n^{(t)}, n\in[N]\setminus\mathcal{D}^{(t)}$, such that
\begin{itemize}
    \item {\bf Correctness: }The message must be additively updated by the increment, i.e.,
    \begin{align}
        \mathbf{W}^{(t+1)}=\mathbf{W}^{(t)}+\boldsymbol{\Delta}^{(t)}.\label{def:up-correct}
    \end{align}
    \item {\bf $X^{(t)}$-security: }The increment must be independent of any $X^{(t)}$ coded increments, $0\leq X^{(t)}\leq N$, i.e., for all $\mathcal{X}\subset [N]\setminus\mathcal{D}^{(t)}$ such that $|\mathcal{X}|=X^{(t)}$,
    \begin{align}
        I(\boldsymbol{\Delta}^{(t)}; (\mathbf{Q}_n^{(t)})_{n\in\mathcal{X}})=0.\label{def:Xsec}
    \end{align}
    \item {\bf Storage transition: }The updated storage at each of the available servers must be a deterministic function of the current storage and the coded increment, i.e., for all $n\in [N]\setminus\mathcal{D}^{(t)}$
    \begin{align}
        H(\mathbf{S}_n^{(t+1)}\mid \mathbf{S}_n^{(t)}, \mathbf{Q}_n^{(t)})=0,\label{def:stotrans}
    \end{align}
    On the other hand, the storage at the dropout servers must left untouched, i.e., for all $n\in\mathcal{D}^{(t)}$,
    \begin{align}\label{eq:uddropstor}
        \mathbf{S}_n^{(t+1)}=\mathbf{S}_n^{(t)}.
    \end{align}
    \item {\bf Independence: }The increment is independent of the current message, and the user has no prior information on the server storage, i.e.,
    \begin{align}
        I(\boldsymbol{\Delta}^{(t)},(\mathbf{Q}_n^{(t)})_{n\in[N]\setminus\mathcal{D}^{(t)}};(\mathbf{S}_n^{(t)})_{n\in[N]})=0.\label{def:ind}
    \end{align}
    
\end{itemize}
For an RDCDS scheme, let us define the update threshold $R_u^{(t)}$ as the minimum number of available servers required by the scheme such that the update operation is feasible at time slot $t$. Besides, the communication efficiency of the update operation is measured by the normalized upload cost $C_u^{(t)}$, defined as
\begin{align}
    C_u^{(t)}=\frac{\sum_{n\in [N]\setminus\mathcal{D}^{(t)}} H(\mathbf{Q}_n^{(t)})}{L}.
\end{align}

\section{Main Results}\label{sec:mainres}
The main result of this work is the complete characterization of the best possible update threshold $R_u^{(t)}$, download cost $C_r^{(t)}$ and upload cost $C_u^{(t)}$ for the problem of RDCDS, as formally stated in the following theorem.
\begin{theorem}\label{thm:main}
{\bf (Converse) }For any RDCDS scheme, at any time slot $t\in\mathbb{N}^*, t>t_0$, the following bounds holds.
\begin{align}
    R_u^{(t)}&\geq N-R_r+\ceil*{K_c}+X^{(t)}\\
    C_r^{(t)}&\geq \frac{N-|\mathcal{D}^{(t)}|}{N-R_r+\ceil*{K_c}-|\mathcal{D}^{(t)}|}\\
    C_u^{(t)}&\geq \frac{N-|\mathcal{D}^{(t)}|}{R_r-X^{(t)}-|\mathcal{D}^{(t)}|}.
\end{align}
{\bf (Achievability) }The RDCDS scheme presented in Section \ref{sec:achiv} achieves the following update threshold, download cost and upload cost at any time slot $t\in\mathbb{N}^*$.
\begin{align}
    R_u^{(t)}&= N-R_r+\ceil*{K_c}+X^{(t)}\\
    C_r^{(t)}&= \frac{N-|\mathcal{D}^{(t)}|}{N-R_r+\ceil*{K_c}-|\mathcal{D}^{(t)}|}\\
    C_u^{(t)}&= \frac{N-|\mathcal{D}^{(t)}|}{R_r-X^{(t)}-|\mathcal{D}^{(t)}|}.
\end{align}
\end{theorem}
\subsection{Observations}
\subsubsection{The Competing $R_r, R_u^{(t)}, C_r^{(t)}$ and $C_u^{(t)}$}\label{sec:competing}
\begin{figure}[!h]
  \centering
  \begin{tikzpicture}[xscale=0.9,yscale=0.9]

    \coordinate (P) at (0, 0);
    \coordinate (Xd) at (2, 0);
    \coordinate (Dr) at (5, 0);
    
    \coordinate (X) at (8.5, 0);
    \coordinate (R) at (10.5, 0);
    \coordinate (NsubR) at (14, 0);
    \coordinate (N) at (17, 0);

    \draw[-latex] (P)--(N) node[thick]{};
    \draw (P) node [above,yshift=7pt,thick]{$0$} -- ++(0, 6pt) ;
    \draw (Dr) node [above,yshift=7pt,thick]{$X^{(t)}+|\mathcal{D}^{(t)}|$} -- ++(0, 6pt);
    \draw (Xd) node [above,yshift=7pt,thick]{$X^{(t)}$} -- ++(0, 6pt) ;
    
    \draw (X) node [above,yshift=7pt,thick]{$R_r-\ceil*{K_c}$} -- ++(0, 6pt) ;
    \draw (R) node [above,yshift=7pt,thick]{$R_r$} -- ++(0, 6pt) ;
    \draw (NsubR) node [above,yshift=7pt,thick]{$N-|\mathcal{D}^{(t)}|$} -- ++(0, 6pt) ;
    \draw (N) node [above,yshift=7pt,thick]{$N$} (0, 6pt) ;

    \draw[dashed] (P) -- ++(0, -4);
    \draw[dashed] (Dr) -- ++(0, -4);
    \draw[dashed] (Xd) -- ++(0, -4);
    \draw[dashed] (X) -- ++(0, -4);
    \draw[dashed] (R) -- ++(0, -4);
    \draw[dashed] (NsubR) -- ++(0, -4);
    \draw[dashed] (N) -- ++(0, -4);

    \draw[color=black,thick,pattern=north west lines,pattern color=red!70!black, rounded corners=0.2em] ($(X)+(0, -0.5)$) rectangle ($(NsubR)+(0, -1.3)$) node[fill=white, pos=.5] {\footnotesize Carry info. during read, $D_1$};

    \draw[color=black,thick,pattern=north west lines,pattern color=blue!70!black, rounded corners=0.2em] ($(Dr)+(0, -1.4)$) rectangle ($(R)+(0, -2.2)$) node[fill=white, pos=.5] {\footnotesize Carry info. during update, $D_2$};

    \draw[color=black,thick,pattern= north west lines,pattern color=green!70!black, rounded corners=0.2em] ($(R)+(0, -2.3)$) rectangle ($(N)+(0, -3.1)$) node[fill=white, pos=.5] {\footnotesize Tolerate read dropouts, $D_3$};

    \draw[color=black,thick,pattern= north west lines,pattern color=yellow!70!black, rounded corners=0.2em] ($(Xd)+(0, -3.2)$) rectangle ($(X)+(0, -4.0)$) node[fill=white, pos=.5] {\footnotesize Tolerate update dropouts, $D_4$};

  \end{tikzpicture}
  \caption{At any time slot $t$, the total of $N$ servers are represented as an $N$ dimensional space, partitioned according to the parameters as shown in the axis. The $4$ horizontal bars, from top to bottom, illustrate the maximum possible number of dimensions that can be exploited to carry desired information during the read operation, the update operation and to tolerate dropout servers during the read operation, the update operation and respectively.}
  \label{fig:dims}
\end{figure}
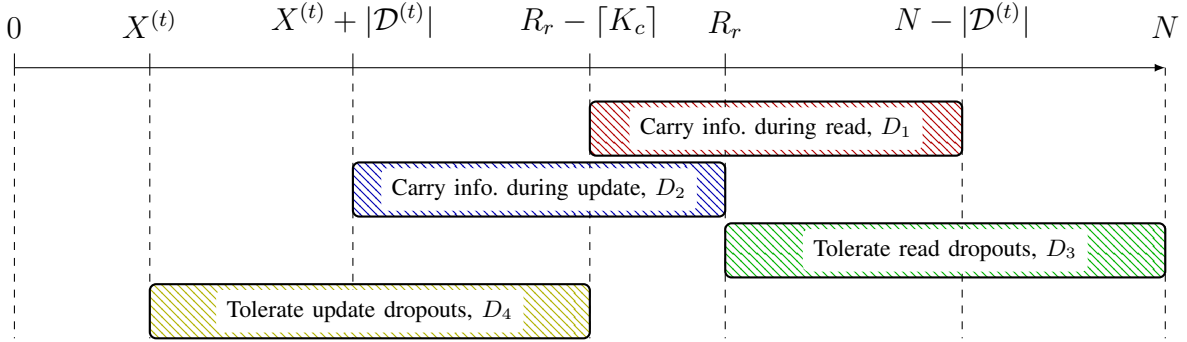

Recall that the results in Theorem \ref{thm:main} are interpreted as the minimum possible update threshold $R_u^{(t)}$, download cost $C_r^{(t)}$, and upload cost $C_u^{(t)}$ for the RDCDS problem. The numerators of the optimal $C_r^{(t)}$ and $C_u^{(t)}$ are both $N-|\mathcal{D}^{(t)}|$, representing the number of available servers at time slot $t$. To further understand the optimality results in Theorem \ref{thm:main}, let us focus on the denominators of the optimal $C_r^{(t)}$ and $C_u^{(t)}$. These denominators represent the maximum number of dimensions (out of the $N-|\mathcal{D}^{(t)}|$ dimensions of available servers) that can be used to carry desired information during the read and update operations. For example, out of $N-|\mathcal{D}^{(t)}|$ dimensions, the maximum number of dimensions that can carry desired information during the read operation is given by $D_1=N-R_r+\ceil*{K_c}-|\mathcal{D}^{(t)}|$. Similarly, during the update operation, this number is $D_2=R_r-X^{(t)}-|\mathcal{D}^{(t)}|$. On the other hand, since the minimum possible $R_u^{(t)}$ is $N-R_r+\ceil*{K_c}+X^{(t)}$, if the user executes the update operation at time slot $t$, the maximum number of dropout servers one can tolerate is $D_4=R_r-\ceil*{K_c}-X^{(t)}$. Similarly, this number for the read operation is $D_3=N-R_r$.

Now, let us conceptually illustrate the total of $N$ servers as an $N$-dimensional space partitioned by the parameters $\ceil*{K_c}$, $R_r$, $X^{(t)}$, and $|\mathcal{D}^{(t)}|$. The quantities $D_1$, $D_2$, $D_3$, and $D_4$ are accordingly represented in Figure \ref{fig:dims}. Since the communication efficiency (the ability to tolerate dropout servers) during the read and update operations improves as $D_1$ and $D_2$ ($D_3$ and $D_4$) increase, respectively, the trade-off among $R_r$, $R_u^{(t)}$, $C_r^{(t)}$, and $C_u^{(t)}$ is now evident from Figure \ref{fig:dims}, which reveals an interesting symmetry in the trade-off among $R_r, R_u^{(t)}, C_r^{(t)}$ and $C_u^{(t)}$. To see this, let us fix the number of servers, the security threshold for coded increment and the number of dropout servers, i.e., $N, X^{(t)}$ and $|\mathcal{D}^{(t)}|$. Then it is clear according to Figure \ref{fig:dims} that the upload cost $C_u^{(t)}$ and the read threshold $R_r$ turn out to be an opposed pair once $R_r-\ceil*{K_c}$ is fixed. On the other hand, we can trade-off $C_r^{(t)}$ and $R_u^{(t)}$ with various $\ceil*{K_c}$ by fixing $R_r$.

It is also of interest to compare the result of Theorem \ref{thm:main} to the ACSA-RW scheme in \cite{Jia_Jafar_XSTPFSL}, where the goal is to allow {\it private} read and update, i.e., the secure storage consists of $K$ messages and any up to $T, T\geq 1$ colluding servers must reveal nothing about the index of the target message (referred to as $T$-privacy). While \cite[Theorem 1]{Jia_Jafar_XSTPFSL} is an achievability result, surprisingly, the read and update thresholds in \cite[Theorem 1]{Jia_Jafar_XSTPFSL} have an analogous form to that of our results. Specifically, the maximum number of dropout servers during the read and the update operation in \cite[Theorem 1]{Jia_Jafar_XSTPFSL} is $N-R_r-T$ and $R_r-K_c-X^{(t)}-T$, respectively (note that $K_c$ in \cite[Theorem 1]{Jia_Jafar_XSTPFSL} must be integer valued). Therefore, it is straightforward to notice the analogy by considering $T$ as the penalty of $T$-privacy. Recall that our result characterizes the maximum possible number of dropout servers, thus this analogy bodes well for the information-theoretic optimality of the threshold values in \cite[Theorem 1]{Jia_Jafar_XSTPFSL}. In terms of communication efficiency, however, the analogy disappears as the achievability scheme in \cite[Theorem 1]{Jia_Jafar_XSTPFSL} fails to exploit the $R_r-(R_r-\ceil*{K_c})=\ceil*{K_c}$ dimensions in the center of the axis to carry desired symbols. Indeed, whether it is possible or not to exploit these dimensions for private read and update is widely open up-to-date as the capacity (i.e., the reciprocal of the minimum possible download cost) of private read operation remains unsolved even for asymptotic settings (i.e., large $K$)\cite{Sun_Jafar_MDSTPIR}. The best achievability result of private read\cite{Jia_Jafar_MDSXSTPIR} fails in this regard and the impossibility is conjectured to be true for asymptotic settings.

\subsubsection{Intuition Behind the Converse Bounds}\label{sec:cvrseintui}
While the formal proof of the converse bounds in Theorem \ref{thm:main} is presented in the form of information-theoretic inequalities in Section \ref{sec:converse} for rigorousness, here let us sketch the proof by explaining the intuition behind it. First of all, the bound on the upload cost $C_u^{(t)}\geq \frac{N-|\mathcal{D}^{(t)}|}{R_r-X^{(t)}-|\mathcal{D}^{(t)}|}$ is perhaps the most intuitive one. Recall that after the update, i.e., at time slot $t+1$, we must be able to recover the message $\mathbf{W}^{(t+1)}$ from {\it any} $R_r$ servers, since the increment $\boldsymbol{\Delta}^{(t)}$ is independent of the current message $\mathbf{W}^{(t)}$ and there are a total of $|\mathcal{D}^{(t)}|$ dropout servers, so in the worst case, one must be able to recover the increment $\boldsymbol{\Delta}^{(t)}$ from {\it any} $R_r-|\mathcal{D}^{(t)}|$ coded increments $\mathbf{Q}_n^{(t)}, n\in[N]\setminus\mathcal{D}^{(t)}$. Therefore, the coded increments must form a threshold secret sharing of the security threshold $X^{(t)}$ and the recovery threshold $R_r-|\mathcal{D}^{(t)}|$, and the bound $C_u^{(t)}\geq \frac{N-|\mathcal{D}^{(t)}|}{R_r-X^{(t)}-|\mathcal{D}^{(t)}|}$ applies according to the standard result of threshold secret sharing (see, e.g., \cite{bitar2017staircase}).

Now let us consider the bound on the update threshold $R_u^{(t)}$, which is essentially explained via the following thought experiment. Assume that the update operation is executed at time slot $t$, and the storage at time slot $t$ at all $N$ servers is made globally known by a genie (so is the current message $\mathbf{W}^{(t)}$). Then at time slot $t+1$, due to the fact that the updated message $\mathbf{W}^{(t)}+\boldsymbol{\Delta}^{(t)}$ must be recoverable from the storage at any $R_r$ servers, and the storage at the dropout servers $\mathcal{D}^{(t)}$ must remain untouched, thus is considered as constant due to the genie, the storage at the servers $[N]\setminus\mathcal{D}^{(t)}$ can be essentially\footnote{The side information, i.e., $\mathbf{S}_n^{(t)}, n\in[N]$, turns out to be useless in terms of reducing the amount of storage required by the secret sharing at the servers $[N]\setminus\mathcal{D}^{(t)}$.} regarded as a threshold secret sharing of the increment $\boldsymbol{\Delta}^{(t)}$, where the security threshold is $X^{(t)}$ and the recovery threshold is $R_r-|\mathcal{D}^{(t)}|$. The standard result of threshold secret sharing shows that the total normalized storage cost of the servers $[N]\setminus\mathcal{D}^{(t)}$ is at least $\frac{N-|\mathcal{D}^{(t)}|}{R_r-X^{(t)}-|\mathcal{D}^{(t)}|}$. On the other hand, according to the storage cost constraint, the total normalized storage cost of the servers $[N]\setminus\mathcal{D}^{(t)}$ is at most $\frac{N-|\mathcal{D}^{(t)}|}{K_c}$, so the number of dropout servers $|\mathcal{D}^{(t)}|$ is at most $R_r-\ceil*{K_c}-X^{(t)}$, from which it applies that $R_u^{(t)}\geq N-R_r+\ceil*{K_c}+X^{(t)}$.

Perhaps the most subtle one is the lower bound on the download cost $C_r^{(t)}$, since the problem setup only requires that the message is recoverable from any $R_r$ servers, which implies only the trivial bound $C_r^{(t)}\geq 1$. Indeed, the lower bound on $C_r^{(t)}$ relies on the lower bound on the update threshold $R_u^{(t)}$ and the assumption that we only consider the steady state of the system, i.e., $t\geq t_0$. First, during the update operations, since the ($X^{(t)}$-securely) coded increment is independent of the current storage and the storage at the dropout servers is untouched, upon the update, for any set $\mathcal{X}\subset[N]$ such that $|\mathcal{X}|=X^{(t)}$ and $\mathcal{D}^{(t)}\cap \mathcal{X}=\emptyset$, the storage at the servers $\mathcal{D}^{(t)}\cup \mathcal{X}$ is independent of the updated message $\mathbf{W}^{(t+1)}=\mathbf{W}^{(t)}+\boldsymbol{\Delta}^{(t)}$. Besides, it turns out that the succeeding read and update operations preserve this independence (which is evident for read operations and for update operations, this is essentially because the user has no prior knowledge on the server storage when generating the coded increments). In other words, if the storage at time slot $t+1$ at the servers $\mathcal{D}^{(t)}\cup \mathcal{X}$ is independent of the message $\mathbf{W}^{(t+1)}$, then for all $\tau>1$, the storage at time slot $t+\tau$ at the servers $\mathcal{D}^{(t)}\cup \mathcal{X}$ is independent of the message $\mathbf{W}^{(t+\tau)}$. Therefore, recall that the maximum number of dropout servers during the update operation is obtained as $R_r-\ceil*{K_c}-X^{(t)}$, once the system enters the steady state, it must have been experienced all possible worst cases in terms of update dropouts, i.e., all possible sets $\mathcal{D}^{(t)}\cup \mathcal{X}\subset[N]$ such that $|\mathcal{D}^{(t)}\cup \mathcal{X}|=R_r-\ceil*{K_c}$. Therefore, for any $t>t_0$, the storage at any up to $R_r-\ceil*{K_c}$ servers is independent of the message $\mathbf{W}^{(t)}$, i.e., the storage at the $N$ servers necessarily form a threshold secret sharing of $\mathbf{W}^{(t)}$ where the security threshold is $X=R_r-\ceil*{K_c}$. The desired bound thus follows from the standard result of threshold secret sharing.

The above argument on the lower bound of the download cost reveals a surprising aspect of the problem of RDCDS, i.e., once the steady state is achieved, the $X=(R_r-\ceil*{K_c})$-security of the storage at the $N$ distributed servers is granted for free. On the one hand, since our focus is on the steady state, from the perspective of the achievability scheme, it is advisable to construct the coded distributed storage that guarantees the $X=(R_r-\ceil*{K_c})$-security over all time slots $t\in\mathbb{N}$. This indeed aligns with our proposed achievability scheme. On the other hand, we note that $R_r-\ceil*{K_c}$ represents the maximum possible value of the security threshold given $K_c$ and $R_r$, as we must have $R_r-X\geq K_c$. Thus, our results also trivially settle the problem of RDCDS with secure storage constraints, i.e., the maximum possible security level is necessarily achieved.

\subsubsection{Comparison to the Staircase Code \cite{bitar2017staircase}}\label{sec:comparesc}
As mentioned, our achievability scheme is inspired by the staircase code \cite{bitar2017staircase}. Intuitively, since in the steady state, the coded distributed storage must satisfy the $(R_r-\ceil*{K_c})$-security, and the scheme must adapt itself to the number of dropout servers to optimize the communication cost during read and update operations, the staircase code for securely coded distributed storage, whose goal is to adaptively achieve the optimal download cost during read operations for all possible number of dropout servers, is considered as a promising starting point. However, recall that the goal of RDCDS is to allow communication-efficient read and update operations simultaneously, it turns out that we cannot treat the two operations separately by solely constructing an update mechanism for the staircase code. This is explained via a motivating example in the following. Consider the setting where $N=4,R_r=2,K_c=1$. According to the construction of the staircase code \cite{bitar2017staircase}, at any time slot $t$, the message $\mathbf{W}^{(t)}$ consists of $L=6$ symbols, and for all $n\in[N]$, the storage is the $n^{th}$ row of the following matrix-product
\begin{align}
    \underbrace{
    \begin{bmatrix}
    1& x_1 &x_1^2 &x_1^3\\
    1& x_2 &x_2^2 &x_2^3\\
    1& x_3 &x_3^2 &x_3^3\\
    1& x_4 &x_4^2 &x_4^3\\
    \end{bmatrix}}_{\mathbf{V}}
    \underbrace{\left[\begin{NiceArray}{cccccc}
        W^{(t)}_1&W^{(t)}_2&\Block[transparent, fill=Plum!20,rounded-corners]{2-1}{}Z_1& \Block[transparent, fill=Cyan!20,rounded-corners]{1-3}{}W^{(t)}_5&W^{(t)}_6&Z_3\\
        W^{(t)}_3&W^{(t)}_4&Z_2&Z_4&Z_5&Z_6\\
        \Block[transparent, fill=Cyan!20,rounded-corners]{1-3}{}W^{(t)}_5&W^{(t)}_6&Z_3&0&0&0\\
        \Block[transparent, fill=Plum!20,rounded-corners]{1-2}{}Z_1&Z_2&0&0&0&0
    \end{NiceArray}\right]}_{\mathbf{M}^{(t)}},
\end{align}
where the matrix $\mathbf{V}$ on the LHS is a Vandermonde matrix, and $Z_i, i\in[6]$ are independent interference/noise symbols. Now assume the user wishes to update the message $\mathbf{W}^{(t)}$ with an increment $\boldsymbol{\Delta}^{(t)}$ to end up with $\mathbf{W}^{(t+1)}=\mathbf{W}^{(t)}+\boldsymbol{\Delta}^{(t)}$. For the sake of simplicity, we assume no dropout servers and set $X^{(t)}=0$, i.e., no security constraint on the coded increment. To preserve the storage structure so that the succeeding users may remain oblivious to history server states, we consider the following construction of the coded increment
\begin{align}
    \underbrace{\begin{bmatrix}
        1& x_1 &x_1^2 &x_1^3\\
    1& x_2 &x_2^2 &x_2^3\\
    1& x_3 &x_3^2 &x_3^3\\
    1& x_4 &x_4^2 &x_4^3\\
    \end{bmatrix}}_{\mathbf{V}}\underbrace{\left[\begin{NiceArray}{cccccc}
        \Delta^{(t)}_1&\Delta^{(t)}_2&\Block[transparent, fill=Plum!20,rounded-corners]{2-1}{}0&\Block[transparent, fill=Cyan!20,rounded-corners]{1-3}{}\Delta^{(t)}_5&\Delta^{(t)}_6&0\\
        \Delta^{(t)}_3&\Delta^{(t)}_4&0&0&0&0\\
        \Block[transparent, fill=Cyan!20,rounded-corners]{1-3}{}\Delta^{(t)}_5&\Delta^{(t)}_6&0&0&0&0\\
        \Block[transparent, fill=Plum!20,rounded-corners]{1-2}{}0&0&0&0&0&0
    \end{NiceArray}\right]}_{\dot{\mathbf{M}}^{(t)}},\label{eq:obs3sc}
\end{align}
where the $n^{th}$ row is uploaded to Server $n$. The updated storage is then simply the summation of the current storage and the received coded increment. While it can be easily seen that the scheme updates the storage correctly as $\mathbf{M}^{(t)}$ and $\dot{\mathbf{M}}^{(t)}$ share the same structure, the normalized upload cost is however $\frac{16}{6}=\frac{8}{3}$ (note that all zero columns produce zero coded symbols that need no upload cost), which does not match the converse bound $C_u^{(t)}\geq \frac{4}{2}=2$ for this case. In fact, a close inspection reveals that the staircase structure, i.e., the matrix $\mathbf{M}^{(t)}$ and $\dot{\mathbf{M}}^{(t)}$, is not optimal in terms of the construction of update mechanism. To understand this, note that since there are no dropout servers and security requirements, the interference/noise symbols in the matrix $\dot{\mathbf{M}}^{(t)}$ of \eqref{eq:obs3sc} are simply zeros. However, the staircase structure fails to completely take advantage to further reduce the communication cost as the message/increment symbols also appear in some of the rightmost few columns of $\mathbf{M}^{(t)}$/$\dot{\mathbf{M}}^{(t)}$. Therefore, our goal now is to construct an alternative staircase structure that restricts the message/increment symbols to the leftmost few columns, illustrated as follows.
\begin{align}
    \mathbf{M}^{(t)}&=
    \left[\begin{NiceArray}{cccccc}
        W^{(t)}_1&W^{(t)}_2&\Block[transparent, fill=Plum!20,rounded-corners]{2-1}{}W^{(t)}_5&\Block[transparent, fill=Cyan!20,rounded-corners]{1-3}{}Z_1&Z_2&Z_3\\
        W^{(t)}_3&W^{(t)}_4&W^{(t)}_6&Z_4&Z_5&Z_6\\
        \Block[transparent, fill=Plum!20,rounded-corners]{1-2}{}W^{(t)}_5&W^{(t)}_6&\Block[transparent, fill=Cyan!20,rounded-corners]{1-1}{}Z_3&0&0&0\\
        \Block[transparent, fill=Cyan!20,rounded-corners]{1-2}{}Z_1&Z_2&0&0&0&0
    \end{NiceArray}\right]\label{eq:obs3newm}\\
    \dot{\mathbf{M}}^{(t)}&=
    \left[\begin{NiceArray}{cccccc}
        \Delta^{(t)}_1&\Delta^{(t)}_2&\Block[transparent, fill=Plum!20,rounded-corners]{2-1}{}\Delta^{(t)}_5&\Block[transparent, fill=Cyan!20,rounded-corners]{1-3}{}0&0&0\\
        \Delta^{(t)}_3&\Delta^{(t)}_4&\Delta^{(t)}_6&0&0&0\\
        \Block[transparent, fill=Plum!20,rounded-corners]{1-2}{}\Delta^{(t)}_5&\Delta^{(t)}_6&\Block[transparent, fill=Cyan!20,rounded-corners]{1-1}{}0&0&0&0\\
        \Block[transparent, fill=Cyan!20,rounded-corners]{1-2}{}0&0&0&0&0&0
    \end{NiceArray}\right].\label{eq:obs3newmd}
\end{align}
Now since the rightmost $3$ columns of $\dot{\mathbf{M}}^{(t)}$ are zeros, the normalized upload cost is calculated as $\frac{12}{6}=2$, which matches the converse bound. The new structure is generalized for arbitrary settings, we refer the readers to Section \ref{sec:achiv} for details.

Indeed, one may notice that this new staircase structure does not preserve decodability when there are read dropout servers if the same Vandermonde encoding matrix is used. This is because a successive interference cancellation decoding strategy is used to recover the message in read operations, and conditioning on previously decoded symbols, the resulting linear system is not necessarily invertible due to non-consecutive powers in the Vandermonde structure. As a workaround, we use Cauchy encoding matrices $\mathbf{C}$ to guarantee the invertibility.

\subsubsection{Comparison to the ACSA-RW Scheme \cite{Jia_Jafar_XSTPFSL}}\label{sec:compareacsa}
Earlier in this section, we compared the performance metrics of the achievability scheme presented in this work with those of the ACSA-RW scheme in \cite{Jia_Jafar_XSTPFSL}. Here, let us explain the primary difference in how each scheme tolerates update dropouts. By and large, the two schemes share the similar idea of exploiting redundant symbols to construct the nullspace such that the coded increment $\mathbf{Q}_n^{(t)}$ for the dropout servers $n\in\mathcal{D}^{(t)}$ is guaranteed to be zero, i.e., the ``update'' of the storage at the dropout servers $\mathbf{S}_n^{(t+1)}=\mathbf{S}_n^{(t)}+\mathbf{Q}_n^{(t)}=\mathbf{S}_n^{(t)}+\mathbf{0}=\mathbf{S}_n^{(t)}$ keeps the storage untouched. Specifically, in the ACSA-RW scheme\cite{Jia_Jafar_XSTPFSL}, the introduction of the nullspace is done by the construction block called {\it null-shaper}, which essentially regards the coded increment as the codewords of an evaluation code (of a rational function), and the {\it null-shaper} is indeed a polynomial to be multiplied with, which evaluates zero at selected points that represent the dropout servers. While the Cauchy encoding matrix used in our achievability scheme can be viewed as an evaluation code, we note that however, the idea of {\it null-shaper} does not apply to the achievability scheme in this work as the staircase structure is not necessarily preserved by the {\it null-shaper}. In this work, the construction of the nullspace is achieved via a recursive strategy that exploits the redundant symbols and the staircase structure. To further elaborate, consider the example in \eqref{eq:obs3newm}, \eqref{eq:obs3newmd}. The complete staircase structure (including redundant symbols that are set to zero for best communication efficiency in \eqref{eq:obs3newmd}) for the increment is as follows.
\begin{align}
    \dot{\mathbf{M}}^{(t)}&=
    \left[\begin{NiceArray}{cccccc}
        \Delta^{(t)}_1&\Delta^{(t)}_2&\Block[transparent, fill=Plum!20,rounded-corners]{2-1}{}\Delta^{(t)}_5&\Block[transparent, fill=Cyan!20,rounded-corners]{1-3}{}H_1&H_2&H_3\\
        \Delta^{(t)}_3&\Delta^{(t)}_4&\Delta^{(t)}_6&H_4&H_5&H_6\\
        \Block[transparent, fill=Plum!20,rounded-corners]{1-2}{}\Delta^{(t)}_5&\Delta^{(t)}_6&\Block[transparent, fill=Cyan!20,rounded-corners]{1-1}{}H_3&0&0&0\\
        \Block[transparent, fill=Cyan!20,rounded-corners]{1-2}{}H_1&H_2&0&0&0&0
    \end{NiceArray}\right]
\end{align}
Now let us assume that there is one dropout server, i.e., $|\mathcal{D}^{(t)}|=1$. Recall that our goal is to construct the nullspace such that the coded increment for the dropout server, $\mathbf{C}(\mathcal{D}^{(t)},:)\dot{\mathbf{M}}^{(t)}$, is zero, by exploiting the redundant symbols $H_1,H_2,\cdots,H_6$. This is equivalent to forcing $\mathbf{C}(\mathcal{D}^{(t)},:)\dot{\mathbf{M}}^{(t)}(:,i)=0$ for all $i=1,2,\cdots, 6$. Starting from the first constraint, i.e., $\mathbf{C}(\mathcal{D}^{(t)},:)\dot{\mathbf{M}}^{(t)}(:,1)=0$, one may notice that this can be achieved by viewing $H_1$ as an unknown and solving the corresponding linear system. The linear system is indeed invertible due to the Cauchy structure. Similarly, $H_2, H_3$ can be solved such that $\mathbf{C}(\mathcal{D}^{(t)},:)\dot{\mathbf{M}}^{(t)}(:,i)=0$ for $i=2,3$. At this point, let us fix $H_1, H_2$ and $H_3$ (i.e., viewed as constants), so that we can correspondingly regard $H_4, H_5$ and $H_6$ as unknowns and solve the linear systems such that $\mathbf{C}(\mathcal{D}^{(t)},:)\dot{\mathbf{M}}^{(t)}(:,i)=0$ for $i=4,5,6$. This is what we refer to as the recursive strategy, and the readers are referred to Section \ref{sec:achiv} for more details.

\section{Proof of Theorem \ref{thm:main}: Converse}\label{sec:converse}
We need the following lemmas to proceed.

\begin{lemma}\label{lemma:increcover}
    Assume that the update operation is executed at time slot $t$. Then for all $\mathcal{R}\subset[N]\setminus\mathcal{D}^{(t)}$ such that $|\mathcal{R}|=R_r-|\mathcal{D}^{(t)}|$, we have $H(\boldsymbol{\Delta}^{(t)}\mid (\mathbf{Q}_n^{(t)})_{n\in\mathcal{R}})=0$.
\end{lemma}
\begin{proof}
    Let us define $\mathcal{R}'=\mathcal{R}\cup\mathcal{D}^{(t)}$. Note that $|\mathcal{R}'|=R_r$, we thus have
    \begin{align}
        &H(\boldsymbol{\Delta}^{(t)}\mid (\mathbf{Q}_n^{(t)})_{n\in\mathcal{R}})\\
        =&H(\boldsymbol{\Delta}^{(t)}\mid (\mathbf{Q}_n^{(t)})_{n\in\mathcal{R}},(\mathbf{S}_n^{(t)})_{n\in[N]})+I(\boldsymbol{\Delta}^{(t)};(\mathbf{S}_n^{(t)})_{n\in[N]}\mid (\mathbf{Q}_n^{(t)})_{n\in\mathcal{R}})\label{eq:lm1-1}\\
        \le&H(\boldsymbol{\Delta}^{(t)}\mid (\mathbf{Q}_n^{(t)})_{n\in\mathcal{R}},(\mathbf{S}_n^{(t)})_{n\in[N]})+I(\boldsymbol{\Delta}^{(t)},(\mathbf{Q}_n^{(t)})_{n\in\mathcal{R}};(\mathbf{S}_n^{(t)})_{n\in[N]} )\label{eq:lm1-2}\\
        =&H(\boldsymbol{\Delta}^{(t)}\mid (\mathbf{Q}_n^{(t)})_{n\in\mathcal{R}},(\mathbf{S}_n^{(t)})_{n\in[N]})\label{eq:lm1-3}\\
        =&H(\mathbf{W}^{(t+1)}\mid (\mathbf{Q}_n^{(t)})_{n\in\mathcal{R}},(\mathbf{S}_n^{(t)})_{n\in[N]})\label{eq:lm1-4}\\
        =&H(\mathbf{W}^{(t+1)}\mid (\mathbf{Q}_n^{(t)})_{n\in\mathcal{R}},(\mathbf{S}_n^{(t)})_{n\in[N]},(\mathbf{S}_n^{(t+1)})_{n\in\mathcal{R}'})\notag\\
        &+I(\mathbf{W}^{(t+1)};(\mathbf{S}_n^{(t+1)})_{n\in\mathcal{D}^{(t)}},(\mathbf{S}_n^{(t+1)})_{n\in\mathcal{R}}\mid (\mathbf{Q}_n^{(t)})_{n\in\mathcal{R}},(\mathbf{S}_n^{(t)})_{n\in[N]})\label{eq:lm1-5}\\
        =&H(\mathbf{W}^{(t+1)}\mid (\mathbf{Q}_n^{(t)})_{n\in\mathcal{R}},(\mathbf{S}_n^{(t)})_{n\in[N]},(\mathbf{S}_n^{(t+1)})_{n\in\mathcal{R}'})\label{eq:lm1-6}\\
        =&0,\label{eq:lm1-7}
    \end{align}
    where \eqref{eq:lm1-1} follows from the definition of mutual information. \eqref{eq:lm1-2} follows from the chain rule and non-negativity of mutual information. \eqref{eq:lm1-3} holds due to the independence constraint \eqref{def:ind}. \eqref{eq:lm1-4} holds because according to the correctness constraint \eqref{def:up-correct}, $\boldsymbol{\Delta}^{(t)}=\mathbf{W}^{(t+1)}-\mathbf{W}^{(t)}$ and $\mathbf{W}^{(t)}$ is fully determined by $(\mathbf{S}_n^{(t)})_{n\in[N]}$. \eqref{eq:lm1-5} is the definition of mutual information, and \eqref{eq:lm1-6} holds due to the fact that $(\mathbf{S}_n^{(t+1)})_{n\in\mathcal{R}}$ is fully determined by $(\mathbf{S}_n^{(t)})_{n\in\mathcal{R}}$ and $(\mathbf{Q}_n^{(t)})_{n\in\mathcal{R}}$ according to the storage transition constraint \eqref{def:stotrans} and the fact that the storage at dropout servers $\mathcal{D}^{(t)}$ must remain untouched by the update operation according to \eqref{eq:uddropstor}, thus $(\mathbf{S}_n^{(t)})_{n\in\mathcal{D}^{(t)}}$ is fully determined by $(\mathbf{S}_n^{(t)})_{n\in[N]}$. Finally, \eqref{eq:lm1-7} follows from the $R_r$-recoverability constraint \eqref{def:Rr-rec}, i.e., $\mathbf{W}^{(t+1)}$ is fully determined by $(\mathbf{S}_n^{(t+1)})_{n\in\mathcal{R}'}$. This completes the proof.
\end{proof}

\begin{lemma}\label{lemma:condindinc}
    Assume that the update operation is executed at time slot $t$. Then for all $\mathcal{X}\subset[N]$ such that $|\mathcal{X}|=X^{(t)}$, we have $I(\boldsymbol{\Delta}^{(t)};(\mathbf{S}_n^{(t+1)})_{n\in\mathcal{X}}\mid (\mathbf{S}_n^{(t)})_{n\in[N]})=0$.
\end{lemma}
\begin{proof}
\begin{align}
    &I(\boldsymbol{\Delta}^{(t)};(\mathbf{S}_n^{(t+1)})_{n\in\mathcal{X}}\mid (\mathbf{S}_n^{(t)})_{n\in[N]})\notag\\
    \le&I(\boldsymbol{\Delta}^{(t)};(\mathbf{S}_n^{(t+1)})_{n\in\mathcal{X}},(\mathbf{Q}_n^{(t)})_{n\in\mathcal{X}}\mid (\mathbf{S}_n^{(t)})_{n\in[N]})\label{eq:lm2-1}\\
    =&I(\boldsymbol{\Delta}^{(t)};(\mathbf{S}_n^{(t+1)})_{n\in\mathcal{X}},\mid (\mathbf{S}_n^{(t)})_{n\in[N]},(\mathbf{Q}_n^{(t)})_{n\in\mathcal{X}})+I(\boldsymbol{\Delta}^{(t)};(\mathbf{Q}_n^{(t)})_{n\in\mathcal{X}}\mid (\mathbf{S}_n^{(t)})_{n\in[N]})\label{eq:lm2-2}\\
    \leq&I(\boldsymbol{\Delta}^{(t)};(\mathbf{Q}_n^{(t)})_{n\in\mathcal{X}}, (\mathbf{S}_n^{(t)})_{n\in[N]})\label{eq:lm2-3}\\
    =&I(\boldsymbol{\Delta}^{(t)};(\mathbf{Q}_n^{(t)})_{n\in\mathcal{X}})+I(\boldsymbol{\Delta}^{(t)}; (\mathbf{S}_n^{(t)})_{n\in[N]}\mid (\mathbf{Q}_n^{(t)})_{n\in\mathcal{X}})\label{eq:lm2-4}\\
    \le&I(\boldsymbol{\Delta}^{(t)},(\mathbf{Q}_n^{(t)})_{n\in\mathcal{X}}; (\mathbf{S}_n^{(t)})_{n\in[N]})\label{eq:lm2-5}\\
    =&0,\label{eq:lm2-6}
\end{align}
where \eqref{eq:lm2-1} holds due to the chain rule and non-negativity of mutual information. \eqref{eq:lm2-2} follows from the chain rule of mutual information. \eqref{eq:lm2-3} holds due to the chain rule and non-negativity of mutual information and according to \eqref{def:stotrans}, $(\mathbf{S}_n^{(t+1)})_{n\in\mathcal{X}}$ is fully determined by $(\mathbf{S}_n^{(t)})_{n\in[N]}$ and $(\mathbf{Q}^{(t)})_{n\in\mathcal{X}}$. Again, \eqref{eq:lm2-4} follows from the chain rule of mutual information. \eqref{eq:lm2-5} holds due to the $X^{(t)}$-security constraint \eqref{def:Xsec} and the chain rule and non-negativity of mutual information. \eqref{eq:lm2-6} follows from the independence constraint \eqref{def:ind}. This completes the proof.
\end{proof}

\begin{lemma}\label{lemma:sdindcondw}
    Assume that the update operation is executed at time slot $t$. Then for all $\mathcal{X}\subset[N]$ such that $|\mathcal{X}|=X^{(t)}, \mathcal{X}\cap\mathcal{D}^{(t)}=\emptyset$, we have\footnote{Recall that according to Lemma \ref{lemma:increcover}, the increment $\boldsymbol{\Delta}^{(t)}$ must be recoverable from any $R_r-|\mathcal{D}^{(t)}|$ coded increments $\mathbf{Q}_n^{(t)}, n\in[N]\setminus\mathcal{D}^{(t)}$. Therefore, we must trivially have $X^{(t)}<R_r-|\mathcal{D}^{(t)}|\leq N-|\mathcal{D}^{(t)}|$, otherwise the recoverability contradicts the $X^{(t)}$-security. The existence of the set $\mathcal{X}$ is thus guaranteed. Similarly, the existence of the set $\mathcal{X}$ in Lemma \ref{lemma:indeinh} and the sets $\mathcal{R}, \mathcal{X}$ in the proof of the lower bound on the update threshold $R_u^{(t)}$ is also guaranteed.} $I\left((\mathbf{S}_n^{(t+1)})_{n\in\mathcal{X}\cup\mathcal{D}^{(t)}}; \boldsymbol{\Delta}^{(t)}\mid \mathbf{W}^{(t)}\right)=0$.
\end{lemma}
\begin{proof}
    \begin{align}
        0&\leq I\left((\mathbf{S}_n^{(t+1)})_{n\in\mathcal{X}\cup\mathcal{D}^{(t)}}; \boldsymbol{\Delta}^{(t)}\mid \mathbf{W}^{(t)}\right)\label{eq:lm3-1}\\
        &=I\left((\mathbf{S}_n^{(t)})_{n\in\mathcal{D}^{(t)}}, (\mathbf{S}_n^{(t+1)})_{n\in\mathcal{X}}; \boldsymbol{\Delta}^{(t)}\mid \mathbf{W}^{(t)}\right)\label{eq:lm3-2}\\
        &=I\left((\mathbf{S}_n^{(t)})_{n\in\mathcal{D}^{(t)}}; \boldsymbol{\Delta}^{(t)}\mid \mathbf{W}^{(t)}\right)+I\left((\mathbf{S}_n^{(t+1)})_{n\in\mathcal{X}}; \boldsymbol{\Delta}^{(t)}\mid (\mathbf{S}_n^{(t)})_{n\in\mathcal{D}^{(t)}},\mathbf{W}^{(t)}\right)\label{eq:lm3-3}\\
        &\leq I\left((\mathbf{S}_n^{(t)})_{n\in\mathcal{D}^{(t)}}, \mathbf{W}^{(t)}; \boldsymbol{\Delta}^{(t)}\right)+I\left((\mathbf{S}_n^{(t+1)})_{n\in\mathcal{X}}; \boldsymbol{\Delta}^{(t)}\mid (\mathbf{S}_n^{(t)})_{n\in\mathcal{D}^{(t)}},\mathbf{W}^{(t)}\right)\label{eq:lm3-4}\\
        &\leq I\left((\mathbf{S}_n^{(t)})_{n\in[N]}; \boldsymbol{\Delta}^{(t)}\right)+I\left((\mathbf{S}_n^{(t+1)})_{n\in\mathcal{X}}; \boldsymbol{\Delta}^{(t)}\mid (\mathbf{S}_n^{(t)})_{n\in\mathcal{D}^{(t)}},\mathbf{W}^{(t)}\right)\label{eq:lm3-5}\\
        &=I\left((\mathbf{S}_n^{(t+1)})_{n\in\mathcal{X}}; \boldsymbol{\Delta}^{(t)}\mid (\mathbf{S}_n^{(t)})_{n\in\mathcal{D}^{(t)}},\mathbf{W}^{(t)}\right)\label{eq:lm3-6}\\
        &\leq I\left((\mathbf{S}_n^{(t+1)})_{n\in\mathcal{X}}, (\mathbf{S}_n^{(t)})_{n\in\mathcal{D}^{(t)}},\mathbf{W}^{(t)}; \boldsymbol{\Delta}^{(t)}\right)\label{eq:lm3-7}\\
        &\leq I\left((\mathbf{Q}_n^{(t)})_{n\in\mathcal{X}}, (\mathbf{S}_n^{(t)})_{n\in[N]}; \boldsymbol{\Delta}^{(t)}\right)\label{eq:lm3-8}\\
        &=I\left((\mathbf{Q}_n^{(t)}\right)_{n\in\mathcal{X}}; \boldsymbol{\Delta}^{(t)})+I\left((\mathbf{S}_n^{(t)})_{n\in[N]}; \boldsymbol{\Delta}^{(t)}\mid (\mathbf{Q}_n^{(t)})_{n\in\mathcal{X}}\right)\label{eq:lm3-9}\\
        &=I\left((\mathbf{S}_n^{(t)})_{n\in[N]}; \boldsymbol{\Delta}^{(t)}\mid (\mathbf{Q}_n^{(t)})_{n\in\mathcal{X}}\right)\label{eq:lm3-10}\\
        &\leq I(\boldsymbol{\Delta}^{(t)}, (\mathbf{Q}_n^{(t)})_{n\in\mathcal{X}};(\mathbf{S}_n^{(t)})_{n\in[N]})\label{eq:lm3-11}\\
        &=0,\label{eq:lm3-12}
    \end{align}
    where \eqref{eq:lm3-1} follows from the non-negativity of mutual information.
    \eqref{eq:lm3-2} holds due to the fact that the storage at the dropout servers $\mathcal{D}^{(t)}$ must remain untouched by the update operation according to \eqref{eq:uddropstor}.
    \eqref{eq:lm3-3} and \eqref{eq:lm3-9} follow from the chain rule of mutual information.
    \eqref{eq:lm3-4}, \eqref{eq:lm3-7} and \eqref{eq:lm3-11} follow from the chain rule and non-negativity of mutual information.
    \eqref{eq:lm3-5} is true because according to the $R_r$-recoverability constraint \eqref{def:Rr-rec}, $\mathbf{W}^{(t)}$ is fully determined by $(\mathbf{S}_n^{(t)})_{n\in[N]}$.
    \eqref{eq:lm3-6} and \eqref{eq:lm3-12} hold due to the independence constraint \eqref{def:ind}.
    \eqref{eq:lm3-8} holds since according to the $R_r$-recoverability constraint \eqref{def:Rr-rec} and the storage transition constraint \eqref{def:stotrans}, $(\mathbf{W}^{(t)}, (\mathbf{S}_n^{(t+1)})_{n\in\mathcal{X}})$ is a deterministic function of $(\mathbf{S}_n^{(t)})_{n\in[N]}$ and  $(\mathbf{Q}_n^{(t)})_{n\in\mathcal{X}}$.
    \eqref{eq:lm3-10} holds due to the $X^{(t)}$-security constraint \eqref{def:Xsec}.
\end{proof}

\begin{lemma}\label{lemma:indeinh}
    Assume that the update operation is executed at time slot $t$. Then for all $\mathcal{X}\subset[N]$ such that $|\mathcal{X}|=X^{(t)}, \mathcal{X}\cap\mathcal{D}^{(t)}=\emptyset$ and all $\tau\in\mathbb{N}^*$, we have $I\left((\mathbf{S}_n^{(t+\tau)})_{n\in\mathcal{X}\cup\mathcal{D}^{(t)}}; \mathbf{W}^{(t+\tau)}\right)=0$.
\end{lemma}
\begin{proof}
    Let us set up a proof by induction. In particular, for the base case, we have
    \begin{align}
        0&\leq I\left((\mathbf{S}_n^{(t+1)})_{n\in\mathcal{X}\cup\mathcal{D}^{(t)}}; \mathbf{W}^{(t+1)}\right)\label{eq:lm4-1}\\
        &\leq I\left((\mathbf{S}_n^{(t+1)})_{n\in\mathcal{X}\cup\mathcal{D}^{(t)}}, \mathbf{W}^{(t)}; \mathbf{W}^{(t+1)}\right)\label{eq:lm4-2}\\
        &=I(\mathbf{W}^{(t)}; \mathbf{W}^{(t)}+\boldsymbol{\Delta}^{(t)})+I\left((\mathbf{S}_n^{(t+1)})_{n\in\mathcal{X}\cup\mathcal{D}^{(t)}}; \mathbf{W}^{(t)}+\boldsymbol{\Delta}^{(t)}\mid \mathbf{W}^{(t)}\right)\label{eq:lm4-3}\\
        &=I\left((\mathbf{S}_n^{(t+1)})_{n\in\mathcal{X}\cup\mathcal{D}^{(t)}}; \boldsymbol{\Delta}^{(t)}\mid \mathbf{W}^{(t)}\right)\label{eq:lm4-4}\\
        &=0,\label{eq:lm4-5}
    \end{align}
    where \eqref{eq:lm4-1} follows from the non-negativity of mutual information.
    \eqref{eq:lm4-2} follows from the chain rule and non-negativity of mutual information.
    \eqref{eq:lm4-3} holds due to the chain rule of mutual information and the correctness constraint \eqref{def:up-correct}.
    \eqref{eq:lm4-4} follows from the independence constraint \eqref{def:ind}, i.e., $\boldsymbol{\Delta}^{(t)}$ is independent of $\mathbf{W}^{(t)}$, and the fact that $H(\boldsymbol{\Delta}^{(t)})=L$.
    \eqref{eq:lm4-5} holds due to Lemma \ref{lemma:sdindcondw}.

    Now for the induction step, let us assume that at time slot $t+\tau, \tau\geq 1$ we have \\$I\left((\mathbf{S}_n^{(t+\tau)})_{n\in\mathcal{X}\cup\mathcal{D}^{(t)}}; \mathbf{W}^{(t+\tau)}\right)=0$, then
    \begin{itemize}
        \item{\bf Case 1: }If the read operation is executed at time slot $t+\tau$, then according to \eqref{def:stortranrd}, we have $\mathbf{S}_n^{(t+\tau)}=\mathbf{S}_n^{(t+\tau+1)}$ for all $n\in[N]$ and $\mathbf{W}^{(t+\tau)}=\mathbf{W}^{(t+\tau+1)}$, thus \\$I\left((\mathbf{S}_n^{(t+\tau+1)})_{n\in\mathcal{X}\cup\mathcal{D}^{(t)}}; \mathbf{W}^{(t+\tau+1)}\right)=0$.
        \item{\bf Case 2: }If the write operation is executed at time slot $t+\tau$, let us define $\mathbf{Q}_n^{(t+\tau)}=\emptyset$ for all $n\in\mathcal{D}^{(t+\tau)}$, then we have
        \begin{align}
            0\leq&I\left((\mathbf{S}_n^{(t+\tau+1)})_{n\in\mathcal{X}\cup\mathcal{D}^{(t)}}; \mathbf{W}^{(t+\tau+1)}\right)\label{eq:lm4indc21}\\
            =&I\left((\mathbf{S}_n^{(t+\tau+1)})_{n\in\mathcal{X}\cup\mathcal{D}^{(t)}}; \mathbf{W}^{(t+\tau)}+\boldsymbol{\Delta}^{(t+\tau)}\right)\label{eq:lm4indc22}\\
            \leq& I\left((\mathbf{S}_n^{(t+\tau+1)})_{n\in\mathcal{X}\cup\mathcal{D}^{(t)}}; \mathbf{W}^{(t+\tau)}+\boldsymbol{\Delta}^{(t+\tau)}\mid (\mathbf{Q}_n^{(t+\tau)})_{n\in[N]}\right)\label{eq:lm4indc23}\\
            =&H\left(\mathbf{W}^{(t+\tau)}+\boldsymbol{\Delta}^{(t+\tau)}\mid (\mathbf{Q}_n^{(t+\tau)})_{n\in[N]}\right)\notag\\
            &-H\left(\mathbf{W}^{(t+\tau)}+\boldsymbol{\Delta}^{(t+\tau)}\mid (\mathbf{S}_n^{(t+\tau+1)})_{n\in\mathcal{X}\cup\mathcal{D}^{(t)}},(\mathbf{Q}_n^{(t+\tau)})_{n\in[N]}\right)\label{eq:lm4indc24}\\
            =&H\left(\mathbf{W}^{(t+\tau)}\mid (\mathbf{Q}_n^{(t+\tau)})_{n\in[N]}\right)-H\left(\mathbf{W}^{(t+\tau)}\mid (\mathbf{S}_n^{(t+\tau+1)})_{n\in\mathcal{X}\cup\mathcal{D}^{(t)}},(\mathbf{Q}_n^{(t+\tau)})_{n\in[N]}\right)\label{eq:lm4indc25}\\
            \leq& H\left(\mathbf{W}^{(t+\tau)}\mid (\mathbf{Q}_n^{(t+\tau)})_{n\in[N]}\right)-H\left(\mathbf{W}^{(t+\tau)}\mid (\mathbf{S}_n^{(t+\tau)})_{n\in\mathcal{X}\cup\mathcal{D}^{(t)}},(\mathbf{Q}_n^{(t+\tau)})_{n\in[N]}\right)\label{eq:lm4indc26}\\
            =& I\left((\mathbf{S}_n^{(t+\tau)})_{n\in\mathcal{X}\cup\mathcal{D}^{(t)}}; \mathbf{W}^{(t+\tau)}\mid (\mathbf{Q}_n^{(t+\tau)})_{n\in[N]}\right)\label{eq:lm4indc27}\\
            \leq& I\left((\mathbf{S}_n^{(t+\tau)})_{n\in\mathcal{X}\cup\mathcal{D}^{(t)}}; \mathbf{W}^{(t+\tau)}, (\mathbf{Q}_n^{(t+\tau)})_{n\in[N]}\right)\label{eq:lm4indc28}\\
            =&I\left((\mathbf{S}_n^{(t+\tau)})_{n\in\mathcal{X}\cup\mathcal{D}^{(t)}}; \mathbf{W}^{(t+\tau)}\right)+I\left((\mathbf{S}_n^{(t+\tau)})_{n\in\mathcal{X}\cup\mathcal{D}^{(t)}}; (\mathbf{Q}_n^{(t+\tau)})_{n\in[N]}\mid \mathbf{W}^{(t+\tau)}\right)\label{eq:lm4indc29}\\
            =&I\left((\mathbf{S}_n^{(t+\tau)})_{n\in\mathcal{X}\cup\mathcal{D}^{(t)}}; (\mathbf{Q}_n^{(t+\tau)})_{n\in[N]}\mid \mathbf{W}^{(t+\tau)}\right)\label{eq:lm4indc210}\\
            \leq& I\left((\mathbf{Q}_n^{(t+\tau)})_{n\in[N]}; (\mathbf{S}_n^{(t+\tau)})_{n\in\mathcal{X}\cup\mathcal{D}^{(t)}},\mathbf{W}^{(t+\tau)}\right)\label{eq:lm4indc211}\\
            =&0,\label{eq:lm4indc212}
        \end{align}
        where \eqref{eq:lm4indc21} follows from the non-negativity of mutual information.
        \eqref{eq:lm4indc22} holds due to the correctness constraint \eqref{def:up-correct}.
        \eqref{eq:lm4indc24} and \eqref{eq:lm4indc27} follow from the definition of mutual information.
        \eqref{eq:lm4indc25} holds because $\boldsymbol{\Delta}^{(t+\tau)}$ is fully determined by $(\mathbf{Q}_n^{(t+\tau)})_{n\in[N]}$ according to Lemma \ref{lemma:increcover}.
        \eqref{eq:lm4indc26} holds due to the fact that conditioning reduces entropy and according to the storage transition constraint \eqref{def:stotrans}, $(\mathbf{S}_n^{(t+\tau+1)})_{n\in\mathcal{X}\cup\mathcal{D}^{(t)}}$ is a deterministic function of $(\mathbf{S}_n^{(t+\tau)})_{n\in\mathcal{X}\cup\mathcal{D}^{(t)}}$ and  $(\mathbf{Q}_n^{(t)})_{n\in[N]}$.
        \eqref{eq:lm4indc28}, \eqref{eq:lm4indc29} and \eqref{eq:lm4indc211} follow from the chain rule and non-negativity of mutual information.
        In \eqref{eq:lm4indc210}, we use the induction hypothesis, and \eqref{eq:lm4indc212} follows from the independence constraint \eqref{def:ind}.
        Finally, \eqref{eq:lm4indc23} is justified as follows. Note that
        \begin{align}
            0\leq&I\left((\mathbf{Q}_n^{(t+\tau)})_{n\in[N]};\mathbf{W}^{(t+\tau)}+\boldsymbol{\Delta}^{(t+\tau)}\right)\notag\\
            =&H(\mathbf{W}^{(t+\tau)}+\boldsymbol{\Delta}^{(t+\tau)})-H(\mathbf{W}^{(t+\tau)}+\boldsymbol{\Delta}^{(t+\tau)}\mid (\mathbf{Q}_n^{(t+\tau)})_{n\in[N]})\label{eq:lm4qindwd1}\\
            =&H(\mathbf{W}^{(t+\tau)}+\boldsymbol{\Delta}^{(t+\tau)})-H(\mathbf{W}^{(t+\tau)}\mid (\mathbf{Q}_n^{(t+\tau)})_{n\in[N]})\label{eq:lm4qindwd2}\\
            =&H(\mathbf{W}^{(t+\tau)}+\boldsymbol{\Delta}^{(t+\tau)})-H(\mathbf{W}^{(t+\tau)})\label{eq:lm4qindwd3}\\
            \leq&L-L=0,\label{eq:lm4qindwd4}
        \end{align}
        where \eqref{eq:lm4qindwd1} is the definition of mutual information. \eqref{eq:lm4qindwd2} holds because $\boldsymbol{\Delta}^{(t+\tau)}$ is fully determined by $(\mathbf{Q}_n^{(t+\tau)})_{n\in[N]}$ according to Lemma \ref{lemma:increcover}. \eqref{eq:lm4qindwd3} follows from the independence constraint \eqref{def:ind}, and finally \eqref{eq:lm4qindwd4} holds since uniform distribution maximizes entropy.
        Therefore, $(\mathbf{Q}_n^{(t+\tau)})_{n\in[N]}$ is independent of $\mathbf{W}^{(t+\tau)}+\boldsymbol{\Delta}^{(t+\tau)}$, hence conditioning increases mutual information.
    \end{itemize}
    The proof is thus completed by induction on $\tau$.
\end{proof}

Now we are ready to finish the converse part of Theorem \ref{thm:main} in the following three proofs.

\begin{proof}{\bf (Lower bound on the update threshold $R_u^{(t)}$)}
Let us assume that the update operation is executed at time slot $t$. Then 
for all $\mathcal{R}, \mathcal{X}$ such that $\mathcal{D}^{(t)}\subset \mathcal{R}\subset[N]$, $\mathcal{X}\subset\mathcal{R}\setminus\mathcal{D}^{(t)}$, $|\mathcal{R}|=R_r, |\mathcal{X}|=X^{(t)}$, we have
    \begin{align}
        L&=H(\boldsymbol{\Delta}^{(t)})\\
        &= H(\boldsymbol{\Delta}^{(t)}\mid (\mathbf{S}_n^{(t)})_{n\in[N]})\label{eq:upthrh-1}\\
        &=I(\boldsymbol{\Delta}^{(t)};(\mathbf{S}_n^{(t+1)})_{n\in\mathcal{R}}\mid (\mathbf{S}_n^{(t)})_{n\in[N]})+H(\boldsymbol{\Delta}^{(t)}\mid (\mathbf{S}_n^{(t+1)})_{n\in\mathcal{R}},(\mathbf{S}_n^{(t)})_{n\in[N]})\label{eq:upthrh-2}\\
        &=I(\boldsymbol{\Delta}^{(t)};(\mathbf{S}_n^{(t+1)})_{n\in\mathcal{R}}\mid (\mathbf{S}_n^{(t)})_{n\in[N]})\label{eq:upthrh-3}\\
        &=I(\boldsymbol{\Delta}^{(t)};(\mathbf{S}_n^{(t+1)})_{n\in\mathcal{X}}\mid (\mathbf{S}_n^{(t)})_{n\in[N]})+I(\boldsymbol{\Delta}^{(t)};(\mathbf{S}_n^{(t+1)})_{n\in\mathcal{R}\setminus\mathcal{X}}\mid (\mathbf{S}_n^{(t)})_{n\in[N]},(\mathbf{S}_n^{(t+1)})_{n\in\mathcal{X}})\label{eq:upthrh-4}\\
        &=I(\boldsymbol{\Delta}^{(t)};(\mathbf{S}_n^{(t+1)})_{n\in\mathcal{R}\setminus\mathcal{X}}\mid (\mathbf{S}_n^{(t)})_{n\in[N]},(\mathbf{S}_n^{(t+1)})_{n\in\mathcal{X}})\label{eq:upthrh-5}\\
        &=I(\boldsymbol{\Delta}^{(t)};(\mathbf{S}_n^{(t)})_{n\in\mathcal{D}^{(t)}},(\mathbf{S}_n^{(t+1)})_{n\in\mathcal{R}\setminus\mathcal{X}\setminus \mathcal{D}^{(t)}}\mid (\mathbf{S}_n^{(t)})_{n\in[N]},(\mathbf{S}_n^{(t+1)})_{n\in\mathcal{X}})\label{eq:upthrh-6}\\
        &=I(\boldsymbol{\Delta}^{(t)};(\mathbf{S}_n^{(t+1)})_{n\in\mathcal{R}\setminus\mathcal{X}\setminus\mathcal{D}^{(t)}}\mid (\mathbf{S}_n^{(t)})_{n\in[N]},(\mathbf{S}_n^{(t+1)})_{n\in\mathcal{X}})\label{eq:upthrh-7}\\
        &\leq H\left((\mathbf{S}_n^{(t+1)})_{n\in\mathcal{R}\setminus\mathcal{X}\setminus\mathcal{D}^{(t)}}\mid (\mathbf{S}_n^{(t)})_{n\in[N]},(\mathbf{S}_n^{(t+1)})_{n\in\mathcal{X}}\right)\label{eq:upthrh-8}\\
        &\leq H\left((\mathbf{S}_n^{(t+1)})_{n\in\mathcal{R}\setminus\mathcal{X}\setminus\mathcal{D}^{(t)}}\right)\label{eq:upthrh-9}\\
        &\leq \sum_{n\in\mathcal{R}\setminus\mathcal{X}\setminus\mathcal{D}^{(t)}}H(\mathbf{S}_n^{(t+1)})\label{eq:upthrh-10}
    \end{align}
    in $q$-ary units. Steps are justified as follows. \eqref{eq:upthrh-1} holds because of the independence constraint \eqref{def:ind}. \eqref{eq:upthrh-2} follows from the definition of mutual information. \eqref{eq:upthrh-3} holds because according to the $R_r$-recoverability constraint \eqref{def:Rr-rec} and correctness constraint \eqref{def:up-correct}, $\mathbf{W}^{(t)}$ and $\mathbf{W}^{(t+1)}$ are fully determined by $(\mathbf{S}_n^{(t+1)})_{n\in\mathcal{R}}$ and $(\mathbf{S}_n^{(t)})_{n\in[N]}$, and $\boldsymbol{\Delta}^{(t)}=\mathbf{W}^{(t+1)}-\mathbf{W}^{(t)}$. \eqref{eq:upthrh-4} follows from the chain rule of mutual information. \eqref{eq:upthrh-5} holds due to Lemma \ref{lemma:condindinc}. Next, \eqref{eq:upthrh-6} follows from the fact that the storage at dropout servers $\mathcal{D}^{(t)}$ must remain untouched by the update operation according to \eqref{eq:uddropstor}. \eqref{eq:upthrh-7} holds since conditioning on $(\mathbf{S}_n^{(t)})_{n\in[N]}$, $(\mathbf{S}_n^{(t)})_{n\in\mathcal{D}^{(t)}}$ is constant. In \eqref{eq:upthrh-8}, we used the definition of mutual information and non-negativity of entropy, and in \eqref{eq:upthrh-9}, \eqref{eq:upthrh-10} we repeatedly used the fact that conditioning reduces entropy.
    
    Averaging the above inequality over all possible choices of $\mathcal{X}$ and $\mathcal{R}$, we have
    \begin{align}
        \frac{\sum_{n\in[N]\setminus\mathcal{D}^{(t)}}H(\mathbf{S}_n^{(t+1)})}{L}\geq \frac{N-|\mathcal{D}^{(t)}|}{R_r-X^{(t)}-|\mathcal{D}^{(t)}|}.
    \end{align}
    On the other hand, according to the constraint of storage cost \eqref{eq:minstor}, we have
    \begin{align}
        \frac{\sum_{n\in[N]\setminus\mathcal{D}^{(t)}}H(\mathbf{S}_n^{(t+1)})}{L}\leq \frac{N-|\mathcal{D}^{(t)}|}{K_c}.
    \end{align}
    Therefore, 
    \begin{align}
        |\mathcal{D}^{(t)}|&\leq R_r-K_c-X^{(t)}.
    \end{align}
    Note that $|\mathcal{D}^{(t)}|\in\mathbb{N}^*$, we have
    \begin{align}
        |\mathcal{D}^{(t)}|&\leq R_r-\ceil*{K_c}-X^{(t)}\label{eq:boundonD}\\
        \Rightarrow R_u^{(t)}&\geq N-R_r+\ceil*{K_c}+X^{(t)}.
    \end{align}
    This completes the proof.
\end{proof}
\begin{proof}{\bf (Lower bound on the download cost $C_r^{(t)}$)}
    Recall that we consider only the steady state of the system, thus for $t>t_0$ the user must have experienced all possible dropout states during a series of read and update operations. Since $|\mathcal{D}^{(t)}|\leq R_r-\ceil*{K_c}-X^{(t)}$ for all $t\in\mathbb{N}^*$ according to \eqref{eq:boundonD}, there must exist time slots $t_1,t_2,\cdots,t_0$ such that the users execute the update operation at these time slots and the sets $\mathcal{D}^{(t)}\cup\mathcal{X}^{(t)}, t=t_1,t_2,\cdots,t_0$ take all possible subsets of $[N]$ such that $|\mathcal{D}^{(t)}\cup\mathcal{X}^{(t)}|= R_r-\ceil*{K_c}$, where $\mathcal{X}^{(t)}$ is an arbitrary subset of $[N]$ such that $\mathcal{X}^{(t)}\cap \mathcal{D}^{(t)}=\emptyset, |\mathcal{X}^{(t)}|=X^{(t)}$. Then, according to Lemma \ref{lemma:indeinh}, for all $t>t_0$ and $\mathcal{X}\subset[N]$ such that $|\mathcal{X}|=R_r-\ceil*{K_c}$, we have 
    \begin{align}\label{eq:scxsec}
        I\left((\mathbf{S}_n^{(t)})_{n\in\mathcal{X}}; \mathbf{W}^{(t)}\right)=0.
    \end{align}
    
    At this point, the lower bound on $C_r^{(t)}$ follows from a quite standard result of threshold secret sharing, see, e.g., \cite{huang2016communication,bitar2017staircase}. However, for the sake of completeness, the following is a proof in our system of notations. Let us assume that the read operation is executed at time slot $t, t>t_0$. Then for all $\mathcal{X}\subset[N]\setminus\mathcal{D}^{(t)}$ such that $|\mathcal{X}|=R_r-\ceil*{K_c}$, we have
    \begin{align}
        L&=H(\mathbf{W}^{(t)})\\
        &=I(\mathbf{W}^{(t)};(\mathbf{A}_n^{(t)})_{n\in[N]\setminus\mathcal{D}^{(t)}})+H(\mathbf{W}^{(t)}\mid (\mathbf{A}_n^{(t)})_{n\in[N]\setminus\mathcal{D}^{(t)}})\label{eq:lbdc-1}\\
        &=I(\mathbf{W}^{(t)};(\mathbf{A}_n^{(t)})_{n\in[N]\setminus\mathcal{D}^{(t)}})\label{eq:lbdc-2}\\
        &=I(\mathbf{W}^{(t)};(\mathbf{A}_n^{(t)})_{n\in\mathcal{X}})+I(\mathbf{W}^{(t)};(\mathbf{A}_n^{(t)})_{n\in[N]\setminus\mathcal{X}\setminus\mathcal{D}^{(t)}}\mid (\mathbf{A}_n^{(t)})_{n\in\mathcal{X}})\label{eq:lbdc-3}\\
        &\leq I(\mathbf{W}^{(t)};(\mathbf{S}_n^{(t)})_{n\in\mathcal{X}})+I(\mathbf{W}^{(t)};(\mathbf{A}_n^{(t)})_{n\in[N]\setminus\mathcal{X}\setminus\mathcal{D}^{(t)}}\mid (\mathbf{A}_n^{(t)})_{n\in\mathcal{X}})\label{eq:lbdc-4}\\
        &=I(\mathbf{W}^{(t)};(\mathbf{A}_n^{(t)})_{n\in[N]\setminus\mathcal{X}\setminus\mathcal{D}^{(t)}}\mid (\mathbf{A}_n^{(t)})_{n\in\mathcal{X}})\label{eq:lbdc-5}\\
        &\leq H\left((\mathbf{A}_n^{(t)})_{n\in[N]\setminus\mathcal{X}\setminus\mathcal{D}^{(t)}}\mid (\mathbf{A}_n^{(t)})_{n\in\mathcal{X}}\right)\label{eq:lbdc-6}\\
        &\leq H\left((\mathbf{A}_n^{(t)})_{n\in[N]\setminus\mathcal{X}\setminus\mathcal{D}^{(t)}}\right)\label{eq:lbdc-7}\\
        &\leq \sum_{n\in[N]\setminus\mathcal{X}\setminus\mathcal{D}^{(t)}}H(\mathbf{A}_n^{(t)})\label{eq:lbdc-8}
    \end{align}
    in $q$-ary units, where \eqref{eq:lbdc-1} follows from the definition of mutual information. \eqref{eq:lbdc-2} holds due to the correctness constraint \eqref{def:dlcrec}. \eqref{eq:lbdc-3} follows from the chain rule of mutual information. \eqref{eq:lbdc-4} holds because $(\mathbf{A}_n^{(t)})_{n\in\mathcal{X}}$ is fully determined by $(\mathbf{S}_n^{(t)})_{n\in\mathcal{X}}$ according to \eqref{def:detmin}. \eqref{eq:lbdc-5} holds due to \eqref{eq:scxsec}. In \eqref{eq:lbdc-6}, we used the definition of mutual information and non-negativity of entropy, and in \eqref{eq:lbdc-7}, \eqref{eq:lbdc-8} we repeatedly used the fact that conditioning reduces entropy.

    Averaging the above inequality over all possible choices of $\mathcal{X}$, we have
    \begin{align}
        C_r^{(t)}=\frac{\sum_{n\in[N]\setminus\mathcal{D}^{(t)}}H(\mathbf{A}_n^{(t)})}{L}\geq \frac{N-|\mathcal{D}^{(t)}|}{N-R_c+\ceil*{K_c}-|\mathcal{D}^{(t)}|}. 
    \end{align}
    This completes the proof.
\end{proof}
\begin{proof}{\bf (Lower bound on the upload cost $C_u^{(t)}$)}
    Let us assume that the update operation is executed at time slot $t$. Then for all $\mathcal{X},\mathcal{R}$ such that $\mathcal{X}\subset\mathcal{R}\subset[N]\setminus\mathcal{D}^{(t)}$ and $|\mathcal{X}|=X^{(t)}, |\mathcal{R}|=R_r-|\mathcal{D}^{(t)}|$, we have
    \begin{align}
        L&=H(\boldsymbol{\Delta}^{(t)})\\
        &=I(\boldsymbol{\Delta}^{(t)};(\mathbf{Q}_n^{(t)})_{n\in\mathcal{R}})+H(\boldsymbol{\Delta}^{(t)}\mid (\mathbf{Q}_n^{(t)})_{n\in\mathcal{R}})\label{eq:lbuc-1}\\
        &=I(\boldsymbol{\Delta}^{(t)};(\mathbf{Q}_n^{(t)})_{n\in\mathcal{R}})\label{eq:lbuc-2}\\
        &=I(\boldsymbol{\Delta}^{(t)};(\mathbf{Q}_n^{(t)})_{n\in\mathcal{X}})+I(\boldsymbol{\Delta}^{(t)};(\mathbf{Q}_n^{(t)})_{n\in\mathcal{R}\setminus\mathcal{X}}\mid (\mathbf{Q}_n^{(t)})_{n\in\mathcal{X}})\label{eq:lbuc-3}\\
        &=I(\boldsymbol{\Delta}^{(t)};(\mathbf{Q}_n^{(t)})_{n\in\mathcal{R}\setminus\mathcal{X}}\mid (\mathbf{Q}_n^{(t)})_{n\in\mathcal{X}})\label{eq:lbuc-4}\\
        &\leq H\left((\mathbf{Q}_n^{(t)})_{n\in\mathcal{R}\setminus\mathcal{X}}\mid (\mathbf{Q}_n^{(t)})_{n\in\mathcal{X}}\right)\label{eq:lbuc-5}\\
        &\leq H\left((\mathbf{Q}_n^{(t)})_{n\in\mathcal{R}\setminus\mathcal{X}}\right)\label{eq:lbuc-6}\\
        &\leq \sum_{n\in\mathcal{R}\setminus\mathcal{X}} H(\mathbf{Q}_n^{(t)})\label{eq:lbuc-7}
    \end{align}
    in $q$-ary units, where \eqref{eq:lbuc-1} follows from the definition of mutual information. \eqref{eq:lbuc-2} holds due to the Lemma \ref{lemma:increcover}. \eqref{eq:lbuc-3} follows from the chain rule of mutual information. \eqref{eq:lbuc-4} holds due to the $X^{(t)}$-security constraint \eqref{def:Xsec}.  In \eqref{eq:lbuc-5}, we used the definition of mutual information and non-negativity of entropy, and in \eqref{eq:lbuc-6}, \eqref{eq:lbuc-7} we repeatedly used the fact that conditioning reduces entropy.

    Averaging over all possible choices of $\mathcal{R}$ and $\mathcal{X}$, we have
    \begin{align}
        C_u^{(t)}=\frac{\sum_{n\in[N]\setminus\mathcal{D}^{(t)}}H(\mathbf{Q}_n^{(t)})}{L}\geq\frac{N-|\mathcal{D}^{(t)}|}{R_r-X^{(t)}-|\mathcal{D}^{(t)}|}.
    \end{align}
    This completes the proof.
\end{proof}

\section{Proof of Theorem \ref{thm:main}: Achievability}\label{sec:achiv}
Throughout this section, we only consider the setting where $K_c\in\mathbb{N}^*$. Otherwise, it suffices to construct the achievability scheme for the setting where the storage factor is $\ceil*{K_c}$ instead, which trivially satisfy the constraint of storage cost as $1/\ceil*{K_c}\leq 1/K_c$. To make the presentation of the general scheme more accessible, let us start with a motivating example.
\subsection{Motivating Example}\label{sec:exam1}
\subsubsection{Construction of the Storage}
Consider the setting where $N=6, R_r=4, K_c=2$, i.e., the collection of the storage at all of the $N$ servers forms an $(K_c=2, R_r=4, N=6)$-coded distributed storage at any time slot $t, t\in\mathbb{N}$. Let $x_1,x_2,\cdots,x_6,f_1,f_2,\cdots,f_6$ be a total of $12$ distinct elements from a finite field $\mathbb{F}_q, q\geq 12$, and let us set $L=12$, i.e., at any time slot $t$, the message $\mathbf{W}^{(t)}=[W_1^{(t)},W_2^{(t)},\cdots,W_{12}^{(t)}]$ consists of $L=12$ symbols from the finite field $\mathbb{F}_q$. The Cauchy matrix generated by $x_1,x_2,\cdots,x_6,f_1,f_2,\cdots,f_6$ is defined as 
\begin{align}
\mathbf{C}=
\begin{bmatrix}
\frac{1}{x_1-f_1}&\frac{1}{x_1-f_2}&\cdots&\frac{1}{x_1-f_6}\\
\frac{1}{x_2-f_1}&\frac{1}{x_2-f_2}&\cdots&\frac{1}{x_2-f_6}\\
\vdots&\vdots&\cdots&\vdots\\
\frac{1}{x_6-f_1}&\frac{1}{x_6-f_2}&\cdots&\frac{1}{x_6-f_6}
\end{bmatrix}.
\end{align}
For all $t\in\mathbb{N}$, let $Z^{(t)}_1,Z^{(t)}_2,\cdots,Z^{(t)}_{12}$ be symbols from the finite field $\mathbb{F}_q$. In particular, for $t=0$, $Z^{(0)}_1,Z^{(0)}_2,\cdots,Z^{(0)}_{12}$ are uniformly i.i.d. over $\mathbb{F}_q$, independent of the (initial) message $\mathbf{W}^{(0)}$. For all $t\in\mathbb{N}$, let us define
\begin{align}
\mathbf{M}^{(t)}=\left[
\begin{NiceArray}{ccc|c|cc}
W^{(t)}_1 & W^{(t)}_2 & W^{(t)}_3 & \Block[transparent, fill=Plum!20,rounded-corners]{3-1}{} {Z^{(t)}_1}  & \Block[transparent, fill=Cyan!20,rounded-corners]{2-2}{} {Z^{(t)}_4}  & {Z^{(t)}_5} \\
W^{(t)}_4 & W^{(t)}_5 & W^{(t)}_6 & {Z^{(t)}_2}  & {Z^{(t)}_6}  & {Z^{(t)}_8}\\
W^{(t)}_7 & W^{(t)}_8 & W^{(t)}_9 & {Z^{(t)}_3}  & {Z^{(t)}_9}  & {Z^{(t)}_{10}} \\
W^{(t)}_{10} & W^{(t)}_{11} & W^{(t)}_{12} & {Z^{(t)}_7}  & {Z^{(t)}_{11}}  & {Z^{(t)}_{12}}\\
\Block[transparent, fill=Plum!20,rounded-corners]{1-3}{}
{Z^{(t)}_1}  & {Z^{(t)}_2}  & {Z^{(t)}_{3}}  &\Block[transparent, fill=Cyan!20,rounded-corners]{1-1}{} {Z^{(t)}_8}  & 0 & 0\\
\Block[transparent, fill=Cyan!20,rounded-corners]{1-3}{}{Z^{(t)}_4}  & {Z^{(t)}_5}  & {Z^{(t)}_6}  & 0 & 0 & 0
\end{NiceArray}\right].\label{eq:ex-init}
\end{align}
Then for our RDCDS scheme, at any time slot $n\in\mathbb{N}$, the storage at Server $n$, denoted as $\mathbf{S}_n^{(t)}$, is
\begin{align}\label{eq:nstor0}
    \mathbf{S}_n^{(t)} = \left(\mathbf{C}\mathbf{M}^{(t)}\right)(n,:),
\end{align}
i.e., the $n^{th}$ row of $\mathbf{C}\mathbf{M}^{(t)}$. Recall that the storage at time slot $t=0$, i.e., $\mathbf{S}_n^{(0)}, n\in[6]$, is initialized {\it a priori}, thus the global coordinator is responsible for the generation of the initial message $\mathbf{W}^{(0)}$, the initial noise symbols $Z^{(0)}_1,Z^{(0)}_2,\cdots,Z^{(0)}_{12}$ and the initial storage at the $N$ servers. Besides, the constraint of minimal storage cost is satisfied, as each server stores a total of $6$ $q$-ary symbols, $H(\mathbf{S}_n^{(t)})=12/2=L/K_c$.

\subsubsection{Execution of the Read Operation}
\paragraph{Example Case 1}
Assume that the user wishes to execute the read operation at time slot $t$, and $\mathcal{D}^{(t)}=\{6\}$, i.e., Server $6$ drops out at this point. To recover the message $\mathbf{W}^{(t)}$, the user downloads the following symbols from available servers.
\begin{align}
    \mathbf{A}_n^{(t)}=\mathbf{S}_n^{(t)}([4])^\mathtt{T}, n\in[5].
\end{align}
Recall that according to the construction of the secure storage, the downloaded symbols can be written in the following matrix form
\begin{align}
&\left[
\mathbf{A}^{(t)}_1~
\mathbf{A}^{(t)}_2~
\mathbf{A}^{(t)}_3~
\mathbf{A}^{(t)}_4~
\mathbf{A}^{(t)}_5
\right]^\mathtt{T}\notag\\
=&\left(\mathbf{C}\mathbf{M}^{(t)}\right)([5],[4])\\
=&\begin{bmatrix}
\frac{1}{x_1-f_1}&\frac{1}{x_1-f_2}&\cdots&\frac{1}{x_1-f_6}\\
\frac{1}{x_2-f_1}&\frac{1}{x_2-f_2}&\cdots&\frac{1}{x_2-f_6}\\
\frac{1}{x_3-f_1}&\frac{1}{x_3-f_2}&\cdots&\frac{1}{x_3-f_6}\\
\frac{1}{x_4-f_1}&\frac{1}{x_4-f_2}&\cdots&\frac{1}{x_4-f_6}\\
\frac{1}{x_5-f_1}&\frac{1}{x_5-f_2}&\cdots&\frac{1}{x_5-f_6}
\end{bmatrix}
\left[ \begin{NiceArray}{ccc|c}
W^{(t)}_1 & W^{(t)}_2 & W^{(t)}_3 & {Z^{(t)}_1}\\
W^{(t)}_4 & W^{(t)}_5 & W^{(t)}_6 & {Z^{(t)}_2}\\
W^{(t)}_7 & W^{(t)}_8 & W^{(t)}_9 & {Z^{(t)}_3}\\
W^{(t)}_{10} & W^{(t)}_{11} & W^{(t)}_{12} & {Z^{(t)}_7} \\
{Z^{(t)}_1}  & {Z^{(t)}_2}  & {Z^{(t)}_{3}}  & {Z^{(t)}_8}\\
{Z^{(t)}_4}  & {Z^{(t)}_5}  & {Z^{(t)}_6}  & 0 
\end{NiceArray}\right].
\end{align}
Therefore, it is evident to see that the $12$ symbols of the desired message $\mathbf{W}^{(t)}$, i.e., $W_1^{(t)},W_2^{(t)},\cdots,W_{12}^{(t)}$, are recoverable by applying a successive interference cancellation decoding strategy. In particular, note that the square Cauchy matrix on the RHS of \eqref{eq:exrecs1} is invertible, we can first decode the following symbols by solving the following linear system
\begin{align}\label{eq:exrecs1}
\left[
{Z^{(t)}_1}~
{Z^{(t)}_2}~
{Z^{(t)}_3}~
{Z^{(t)}_7}~
{Z^{(t)}_8}~
\right]^\mathtt{T}=
\mathbf{C}([5],[5])^{-1}
\left[
\mathbf{A}^{(t)}_1(4)~
\mathbf{A}^{(t)}_2(4)~
\mathbf{A}^{(t)}_3(4)~
\mathbf{A}^{(t)}_4(4)~
\mathbf{A}^{(t)}_5(4)
\right]^\mathtt{T}.
\end{align}
Then, subtracting the decoded symbols $Z^{(t)}_1,Z^{(t)}_2,Z^{(t)}_3$ from $\mathbf{A}^{(t)}_n(1),\mathbf{A}^{(t)}_n(2)$ and $\mathbf{A}^{(t)}_n(3)$ for all $n\in[5]$, the message is thus recoverable.
\begin{align}
\begin{bmatrix}
W^{(t)}_1 & W^{(t)}_2 & W^{(t)}_3\\
W^{(t)}_4 & W^{(t)}_5 & W^{(t)}_6\\
W^{(t)}_7 & W^{(t)}_8 & W^{(t)}_9\\
W^{(t)}_{10} & W^{(t)}_{11} & W^{(t)}_{12}\\
{Z^{(t)}_4}  & {Z^{(t)}_5}  & {Z^{(t)}_{6}}
\end{bmatrix}
=&\mathbf{C}([5],\{1,2,3,4,6\})^{-1}\notag\\
&\times\begin{bmatrix}
\mathbf{A}^{(t)}_1(1)-\frac{1}{x_1-f_5}{Z^{(t)}_{1}}  
&\mathbf{A}^{(t)}_1(2)-\frac{1}{x_1-f_5}{Z^{(t)}_{2}}  
&\mathbf{A}^{(t)}_1(3)-\frac{1}{x_1-f_5}{Z^{(t)}_{3}}\\
\mathbf{A}^{(t)}_2(1)-\frac{1}{x_2-f_5}{Z^{(t)}_{1}}  
&\mathbf{A}^{(t)}_2(2)-\frac{1}{x_2-f_5}{Z^{(t)}_{2}} 
&\mathbf{A}^{(t)}_2(3)-\frac{1}{x_2-f_5}{Z^{(t)}_{3}} \\
\mathbf{A}^{(t)}_3(1)-\frac{1}{x_3-f_5}{Z^{(t)}_{1}}  
&\mathbf{A}^{(t)}_3(2)-\frac{1}{x_3-f_5}{Z^{(t)}_{2}} 
&\mathbf{A}^{(t)}_3(3)-\frac{1}{x_3-f_5}{Z^{(t)}_{3}} \\
\mathbf{A}^{(t)}_4(1)-\frac{1}{x_4-f_5}{Z^{(t)}_{1}}  
& \mathbf{A}^{(t)}_4(2)-\frac{1}{x_4-f_5}{Z^{(t)}_{2}} 
& \mathbf{A}^{(t)}_4(3)-\frac{1}{x_4-f_5}{Z^{(t)}_{3}} \\
\mathbf{A}^{(t)}_5(1)-\frac{1}{x_5-f_5}{Z^{(t)}_{1}} 
& \mathbf{A}^{(t)}_5(2)-\frac{1}{x_5-f_5}{Z^{(t)}_{2}} 
& \mathbf{A}^{(t)}_5(3)-\frac{1}{x_5-f_5}{Z^{(t)}_{3}}
\end{bmatrix}.
\end{align}
Since a total of $12$ symbols of the desired message is recovered from a total of $5\times 4=20$ downloaded symbols, the normalized download cost is $20/12=5/3$, which matches the converse bound $(N-|\mathcal{D}^{(t)}|)/(N-R_r+K_c-|\mathcal{D}^{(t)}|)=5/3$.

\paragraph{Example Case 2}
Let us consider one more case where at time slot $t$, we have $\mathcal{D}^{(t)}=\{3,6\}$, i.e., Server $3$ and Server $6$ drop out, and the user wishes to recover the message $\mathbf{W}^{(t)}$. To this end, the user downloads the following symbols from available servers $[N]\setminus\mathcal{D}^{(t)}$.
\begin{align}
    \mathbf{A}_n^{(t)}=(\mathbf{S}_n^{(t)})^\mathtt{T}, n\in[N]\setminus\mathcal{D}^{(t)}.
\end{align}
Similarly, since the downloaded symbols at this point can be represented in the following form,
\begin{align}
\left[
\mathbf{A}^{(t)}_1~
\mathbf{A}^{(t)}_2~
\mathbf{A}^{(t)}_4~
\mathbf{A}^{(t)}_5
\right]^\mathtt{T}
=&\mathbf{C}([N]\setminus\mathcal{D}^{(t)}, :)\mathbf{M}^{(t)}\\
=&\begin{bmatrix}
\frac{1}{x_1-f_1}&\frac{1}{x_1-f_2}&\cdots&\frac{1}{x_1-f_6}\\
\frac{1}{x_2-f_1}&\frac{1}{x_2-f_2}&\cdots&\frac{1}{x_2-f_6}\\
\frac{1}{x_4-f_1}&\frac{1}{x_4-f_2}&\cdots&\frac{1}{x_4-f_6}\\
\frac{1}{x_5-f_1}&\frac{1}{x_5-f_2}&\cdots&\frac{1}{x_5-f_6}
\end{bmatrix}
\left[ \begin{NiceArray}{ccc|c|cc}
W^{(t)}_1 & W^{(t)}_2 & W^{(t)}_3 & {Z^{(t)}_1}  & {Z^{(t)}_4}  & {Z^{(t)}_5} \\
W^{(t)}_4 & W^{(t)}_5 & W^{(t)}_6 & {Z^{(t)}_2}  & {Z^{(t)}_6}  & {Z^{(t)}_8}\\
W^{(t)}_7 & W^{(t)}_8 & W^{(t)}_9 & {Z^{(t)}_3}  & {Z^{(t)}_9}  & {Z^{(t)}_{10}} \\
W^{(t)}_{10} & W^{(t)}_{11} & W^{(t)}_{12} & {Z^{(t)}_7}  & {Z^{(t)}_{11}}  & {Z^{(t)}_{12}}\\
{Z^{(t)}_1}  & {Z^{(t)}_2}  & {Z^{(t)}_{3}}  & {Z^{(t)}_8}  & 0 & 0\\
{Z^{(t)}_4}  & {Z^{(t)}_5}  & {Z^{(t)}_6}  & 0 & 0 & 0
\end{NiceArray}\right],
\end{align}
the desired symbols can be recovered via a successive interference cancellation decoding strategy, omitted here for brevity. In terms of the communication efficiency, we now have
\begin{align}
    C_r^{(t)}=\frac{4\times 6}{12}=2,
\end{align}
which also matches the converse bound $(N-|\mathcal{D}^{(t)}|)/(N-R_r+K_c-|\mathcal{D}^{(t)}|)=2$.

\subsubsection{Execution of the Update Operation}
\paragraph{Example Case 1}
Let us assume that the user wishes to update the message to the distributed servers with the generated increment $\mathbf{\Delta}^{(t)}=[\Delta^{(t)}_1,\Delta^{(t)}_2,\cdots,\Delta^{(t)}_{12}]$ at time slot $t$. Let us say $\mathcal{D}^{(t)}=\{5\}$ and $X^{(t)}=0$, i.e., Server $5$ drops out at this time, and no security guarantee is required for this update. To construct the coded increment for available servers, let us define
\begin{align}\label{eq:ex-updM1}
\dot{\mathbf{M}}^{(t)}=
\left[ \begin{NiceArray}{ccc|c|cc}
\Delta^{(t)}_1 & \Delta^{(t)}_2 & \Delta^{(t)}_3 &\Block[transparent, fill=Plum!20,rounded-corners]{3-1}{} {H^{(t)}_1} &\Block[transparent, fill=Cyan!20,rounded-corners]{2-2}{} 0 & 0 \\
\Delta^{(t)}_4 & \Delta^{(t)}_5 & \Delta^{(t)}_6 & {H^{(t)}_2}& 0 & 0 \\
\Delta^{(t)}_7 & \Delta^{(t)}_8 & \Delta^{(t)}_9 & {H^{(t)}_3} & 0 & 0 \\
\Delta^{(t)}_{10} & \Delta^{(t)}_{11} & \Delta^{(t)}_{12} & H^{(t)}_4 & 0 & 0 \\
\Block[transparent, fill=Plum!20,rounded-corners]{1-3}{}{H^{(t)}_1} & {H^{(t)}_2} & {H^{(t)}_3}  &\Block[transparent, fill=Cyan!20,rounded-corners]{1-1}{0} & 0 & 0 \\
\Block[transparent, fill=Cyan!20,rounded-corners]{1-3}{}{0}  & 0 & {0} & {0} & 0 & 0 
\end{NiceArray}\right],
\end{align}
where  
\begin{align}
    &H^{(t)}_1=-\left(x_5-f_5\right)\left(\frac{1}{x_5-f_1}\Delta^{(t)}_1+\frac{1}{x_5-f_2}\Delta^{(t)}_4+\frac{1}{x_5-f_3}\Delta^{(t)}_7+\frac{1}{x_5-f_4}\Delta^{(t)}_{10}\right),\\
    &H^{(t)}_2=-\left(x_5-f_5\right)\left(\frac{1}{x_5-f_1}\Delta^{(t)}_2+\frac{1}{x_5-f_2}\Delta^{(t)}_5+\frac{1}{x_5-f_3}\Delta^{(t)}_8+\frac{1}{x_5-f_4}\Delta^{(t)}_{11}\right),\\
    &H^{(t)}_3=-\left(x_5-f_5\right)\left(\frac{1}{x_5-f_1}\Delta^{(t)}_3+\frac{1}{x_5-f_2}\Delta^{(t)}_6+\frac{1}{x_5-f_3}\Delta^{(t)}_9+\frac{1}{x_5-f_4}\Delta^{(t)}_{12}\right),
\end{align}
and
\begin{align}
    H^{(t)}_4=-\left(x_5-f_4\right)\left(\frac{1}{x_5-f_1}H^{(t)}_1+\frac{1}{x_5-f_2}H^{(t)}_2+\frac{1}{x_5-f_3}H^{(t)}_3\right).
\end{align}
Then, the coded increment $\mathbf{Q}^{(t)}_n, n\in[N]\setminus\mathcal{D}^{(t)}$ is constructed as follows
\begin{align}
    \mathbf{Q}^{(t)}_n = \left(\mathbf{C}\dot{\mathbf{M}}^{(t)}\right)(n, :),
\end{align}
i.e., the $n^{th}$ row of $\mathbf{C}\dot{\mathbf{M}}^{(t)}$. Note that since $|\mathcal{D}^{(t)}|$ is viewed as a constant during the time slot $t$ and $\mathbf{Q}^{(t)}_n(\{5,6\})$ must be zeros, it suffices to upload the first $4$ elements of $\mathbf{Q}^{(t)}_n$ to each of the available servers.

Upon receiving the coded increment, Server $n, n\in[N]\setminus\mathcal{D}^{(t)}$ updates its storage according to the following equation
\begin{align}
    \mathbf{S}_n^{(t+1)}=&\mathbf{S}^{(t)}_n+\mathbf{Q}^{(t)}_n.
\end{align}
Recall that for all $n\in\mathcal{D}^{(t)}$, $\mathbf{S}_n^{(t+1)}=\mathbf{S}^{(t)}_n$ by definition \eqref{eq:uddropstor}. To see the correctness of our construction, it suffices to show that the updated storage $\mathbf{S}_n^{(t+1)}, n\in[N]$, i.e., the storage at time slot $t+1$, is of the same form as in \eqref{eq:nstor0}. The key observation in this regard is that, according to the definition of $H_1^{(t)}, H_2^{(t)}, H_3^{(t)}, H_4^{(t)}$, it is guaranteed that
\begin{align}
    \mathbf{C}(\mathcal{D}^{(t)},:)\dot{\mathbf{M}}^{(t)}=\mathbf{0}.
\end{align}
For example,
\begin{align}\label{eq:examc1H}
    &\left(\mathbf{C}(\mathcal{D}^{(t)},:)\dot{\mathbf{M}}^{(t)}\right)(1,1)\notag\\
    =&\frac{1}{x_5-f_1}\Delta^{(t)}_1+\frac{1}{x_5-f_2}\Delta^{(t)}_4+\frac{1}{x_5-f_3}\Delta^{(t)}_7+\frac{1}{x_5-f_4}\Delta^{(t)}_{10}+\frac{1}{x_5-f_5}H_1^{(t)}\\
    =&0.
\end{align}
Therefore, even though the dropout servers cannot be updated, i.e., $\mathbf{S}_n^{(t+1)}=\mathbf{S}^{(t)}_n$ for all $n\in\mathcal{D}^{(t)}$, at time slot $t+1$, the storage at Server $n, n\in[N]$ can be written uniformly as follows
\begin{align}
\mathbf{S}_n^{(t+1)}
=&\left(\mathbf{C}\mathbf{M}^{(t)}\right)(n,:)+\left(\mathbf{C}\dot{\mathbf{M}}^{(t)}\right)(n,:)\\
=&\mathbf{C}(n,:)\left(
\left[ \begin{NiceArray}{ccc|c|cc}
W^{(t)}_1 & W^{(t)}_2 & W^{(t)}_3 & \Block[transparent, fill=Plum!20,rounded-corners]{3-1}{} {Z^{(t)}_1}  & \Block[transparent, fill=Cyan!20,rounded-corners]{2-2}{} {Z^{(t)}_4}  & {Z^{(t)}_5} \\
W^{(t)}_4 & W^{(t)}_5 & W^{(t)}_6 & {Z^{(t)}_2}  & {Z^{(t)}_6}  & {Z^{(t)}_8}\\
W^{(t)}_7 & W^{(t)}_8 & W^{(t)}_9 & {Z^{(t)}_3}  & {Z^{(t)}_9}  & {Z^{(t)}_{10}} \\
W^{(t)}_{10} & W^{(t)}_{11} & W^{(t)}_{12} & {Z^{(t)}_7}  & {Z^{(t)}_{11}}  & {Z^{(t)}_{12}}\\
\Block[transparent, fill=Plum!20,rounded-corners]{1-3}{}
{Z^{(t)}_1}  & {Z^{(t)}_2}  & {Z^{(t)}_{3}}  &\Block[transparent, fill=Cyan!20,rounded-corners]{1-1}{} {Z^{(t)}_8}  & 0 & 0\\
\Block[transparent, fill=Cyan!20,rounded-corners]{1-3}{}{Z^{(t)}_4}  & {Z^{(t)}_5}  & {Z^{(t)}_6}  & 0 & 0 & 0
\end{NiceArray}\right]+
\left[ \begin{NiceArray}{ccc|c|cc}
\Delta^{(t)}_1 & \Delta^{(t)}_2 & \Delta^{(t)}_3 &\Block[transparent, fill=Plum!20,rounded-corners]{3-1}{} {H^{(t)}_1} &\Block[transparent, fill=Cyan!20,rounded-corners]{2-2}{} 0 & 0 \\
\Delta^{(t)}_4 & \Delta^{(t)}_5 & \Delta^{(t)}_6 & {H^{(t)}_2}& 0 & 0 \\
\Delta^{(t)}_7 & \Delta^{(t)}_8 & \Delta^{(t)}_9 & {H^{(t)}_3} & 0 & 0 \\
\Delta^{(t)}_{10} & \Delta^{(t)}_{11} & \Delta^{(t)}_{12} & H^{(t)}_4 & 0 & 0 \\
\Block[transparent, fill=Plum!20,rounded-corners]{1-3}{}{H^{(t)}_1} & {H^{(t)}_2} & {H^{(t)}_3}  &\Block[transparent, fill=Cyan!20,rounded-corners]{1-1}{0} & 0 & 0 \\
\Block[transparent, fill=Cyan!20,rounded-corners]{1-3}{}{0}  & 0 & {0} & {0} & 0 & 0 
\end{NiceArray}\right]\right)\\
=&\notag\mathbf{C}(n,:)
\left[ \begin{NiceArray}{ccc|c|cc}
W^{(t)}_1+\Delta^{(t)}_1 & W^{(t)}_2+\Delta^{(t)}_2 & W^{(t)}_3+\Delta^{(t)}_3 &\Block[transparent, fill=Plum!20,rounded-corners]{3-1}{} {Z^{(t)}_1+H^{(t)}_1}  &\Block[transparent, fill=Cyan!20,rounded-corners]{2-2}{} {Z^{(t)}_4}  & {Z^{(t)}_5} \\
W^{(t)}_4+\Delta^{(t)}_4 & W^{(t)}_5+\Delta^{(t)}_5 & W^{(t)}_6+\Delta^{(t)}_6 & {Z^{(t)}_2+H^{(t)}_2}  & {Z^{(t)}_6}  & {Z^{(t)}_8}\\
W^{(t)}_7+\Delta^{(t)}_7 & W^{(t)}_8+\Delta^{(t)}_8 & W^{(t)}_9+\Delta^{(t)}_9 & {Z^{(t)}_3+H^{(t)}_3}  & {Z^{(t)}_9}  & {Z^{(t)}_{10}} \\
W^{(t)}_{10}+\Delta^{(t)}_{10} & W^{(t)}_{11}+\Delta^{(t)}_{11} & W^{(t)}_{12}+\Delta^{(t)}_{12} & {Z^{(t)}_7+H^{(t)}_4}  & {Z^{(t)}_{11}}  & {Z^{(t)}_{12}}\\
\Block[transparent, fill=Plum!20,rounded-corners]{1-3}{}{Z^{(t)}_1+H^{(t)}_1}  & {Z^{(t)}_2+H^{(t)}_2}  & {Z^{(t)}_{3}+H^{(t)}_3}  &\Block[transparent, fill=Cyan!20,rounded-corners]{1-1}{} {Z^{(t)}_8}  & 0 & 0\\
\Block[transparent, fill=Cyan!20,rounded-corners]{1-3}{}{Z^{(t)}_4}  & {Z^{(t)}_5}  & {Z^{(t)}_6}  & 0 & 0 & 0
\end{NiceArray}\right]\\
=&\mathbf{C}(n,:)
\left[ \begin{NiceArray}{ccc|c|cc}
W^{(t+1)}_1 & W^{(t+1)}_2 & W^{(t+1)}_3 &\Block[transparent, fill=Plum!20,rounded-corners]{3-1}{} {Z^{(t+1)}_1}  &\Block[transparent, fill=Cyan!20,rounded-corners]{2-2}{} {Z^{(t+1)}_4}  & {Z^{(t+1)}_5} \\
W^{(t+1)}_4 & W^{(t+1)}_5 & W^{(t+1)}_6 & {Z^{(t+1)}_2}  & {Z^{(t+1)}_6}  & {Z^{(t+1)}_8}\\
W^{(t+1)}_7 & W^{(t+1)}_8 & W^{(t+1)}_9 & {Z^{(t+1)}_3}  & {Z^{(t+1)}_9}  & {Z^{(t+1)}_{10}} \\
W^{(t+1)}_{10} & W^{(t+1)}_{11} & W^{(t+1)}_{12} & {Z^{(t+1)}_7}  & {Z^{(t+1)}_{11}}  & {Z^{(t+1)}_{12}}\\
\Block[transparent, fill=Plum!20,rounded-corners]{1-3}{}{Z^{(t+1)}_1}  & {Z^{(t+1)}_2}  & {Z^{(t+1)}_{3}}  &\Block[transparent, fill=Cyan!20,rounded-corners]{1-1}{} {Z^{(t+1)}_8}  & 0 & 0\\
\Block[transparent, fill=Cyan!20,rounded-corners]{1-3}{}{Z^{(t+1)}_4}  & {Z^{(t+1)}_5}  & {Z^{(t+1)}_6}  & 0 & 0 & 0
\end{NiceArray}\right]\\
=&\mathbf{C}(n,:)\mathbf{M}^{(t+1)},
\end{align}
i.e., the storage at time slot $t+1$ is indeed in accordance with \eqref{eq:nstor0}, where the following recurrence relation holds
\begin{align}
    Z_i^{(t+1)}&=Z_i^{(t)}+H_{i}^{(t)}, &i&=1,2,3\\
    Z_7^{(t+1)}&=Z_7^{(t)}+H_{4}^{(t)},&&\\
    Z_i^{(t+1)}&=Z_i^{(t)}, &i&=4,5,6,8,9,10,11,12.
\end{align}
Finally, we can calculate that the normalized upload cost is 
\begin{align}
    C_u^{(t)}=\frac{5\times 4}{12}=\frac{5}{3},
\end{align}
which achieves the converse bound $(N-|\mathcal{D}^{(t)}|)/(R_r-X^{(t)}-|\mathcal{D}^{(t)}|)=5/3$.

\paragraph{Example Case 2}
Let us consider another update operation case where $\mathcal{D}^{(t)}=\{5\}$ and $X^{(t)}=1$, i.e., Server $5$ drops out at this time, and any single server can not reveal any information about the increment $\mathbf{\Delta}^{(t)}$. In order to construct the coded increment, let us define
\begin{align}\label{eq:ex-updM2}
\dot{\mathbf{M}}^{(t)}=
\left[ \begin{NiceArray}{ccc|c|cc}
\Delta^{(t)}_1 & \Delta^{(t)}_2 & \Delta^{(t)}_3 &\Block[transparent, fill=Plum!20,rounded-corners]{3-1}{} {\dot{Z}^{(t)}_1} &\Block[transparent, fill=Cyan!20,rounded-corners]{2-2}{} {H^{(t)}_1} & {H^{(t)}_2}\\
\Delta^{(t)}_4 & \Delta^{(t)}_5 & \Delta^{(t)}_6 & {\dot{Z}^{(t)}_2} & {H^{(t)}_3} & {H^{(t)}_4}\\
\Delta^{(t)}_7 & \Delta^{(t)}_8 & \Delta^{(t)}_9 & {\dot{Z}^{(t)}_3} & {\dot{Z}^{(t)}_9}  & {\dot{Z}^{(t)}_{10}}\\
\Delta^{(t)}_{10} & \Delta^{(t)}_{11} & \Delta^{(t)}_{12} & {\dot{Z}^{(t)}_7}  & {H^{(t)}_5}  & {H^{(t)}_6}\\
\Block[transparent, fill=Plum!20,rounded-corners]{1-3}{}{\dot{Z}^{(t)}_1}  & {\dot{Z}^{(t)}_2} & {\dot{Z}^{(t)}_3} &\Block[transparent, fill=Cyan!20,rounded-corners]{1-1}{} {H^{(t)}_4} & 0 & 0\\
\Block[transparent, fill=Cyan!20,rounded-corners]{1-3}{}{H^{(t)}_1} & {H^{(t)}_2} & {H^{(t)}_3}  & 0 & 0 & 0
\end{NiceArray}\right],
\end{align}
where $\dot{Z}^{(t)}_1,\dot{Z}^{(t)}_2,\dot{Z}^{(t)}_3,\dot{Z}^{(t)}_7,\dot{Z}^{(t)}_9,\dot{Z}^{(t)}_{10}$ are uniformly i.i.d. symbols from $\mathbb{F}_q$, independent of the increment $\mathbf{\Delta}^{(t)}$, used to guarantee $X^{(t)}=1$-security, and
\begin{align}
    &H^{(t)}_1=-\left(x_5-f_6\right)\left(\frac{1}{x_5-f_1}\Delta^{(t)}_1+\frac{1}{x_5-f_2}\Delta^{(t)}_4+\frac{1}{x_5-f_3}\Delta^{(t)}_7+\frac{1}{x_5-f_4}\Delta^{(t)}_{10}+\frac{1}{x_5-f_5}\dot{Z}^{(t)}_1\right),\\
    &H^{(t)}_2=-\left(x_5-f_6\right)\left(\frac{1}{x_5-f_1}\Delta^{(t)}_2+\frac{1}{x_5-f_2}\Delta^{(t)}_5+\frac{1}{x_5-f_3}\Delta^{(t)}_8+\frac{1}{x_5-f_4}\Delta^{(t)}_{11}+\frac{1}{x_5-f_5}\dot{Z}^{(t)}_2\right),\\
    &H^{(t)}_3=-\left(x_5-f_6\right)\left(\frac{1}{x_5-f_1}\Delta^{(t)}_3+\frac{1}{x_5-f_2}\Delta^{(t)}_6+\frac{1}{x_5-f_3}\Delta^{(t)}_9+\frac{1}{x_5-f_4}\Delta^{(t)}_{12}+\frac{1}{x_5-f_5}\dot{Z}^{(t)}_3\right),
\end{align}
and
\begin{align}
    H^{(t)}_4=-\left(x_5-f_5\right)\left(\frac{1}{x_5-f_1}\dot{Z}^{(t)}_1+\frac{1}{x_5-f_2}\dot{Z}^{(t)}_2+\frac{1}{x_5-f_3}\dot{Z}^{(t)}_3+\frac{1}{x_5-f_4}\dot{Z}^{(t)}_7\right),
\end{align}
and
\begin{align}
    &H^{(t)}_5=-\left(x_5-f_4\right)\left(\frac{1}{x_5-f_1}H^{(t)}_1+\frac{1}{x_5-f_2}H^{(t)}_3+\frac{1}{x_5-f_3}\dot{Z}^{(t)}_9\right),\\
    &H^{(t)}_6=-\left(x_5-f_4\right)\left(\frac{1}{x_5-f_1}H^{(t)}_2+\frac{1}{x_5-f_2}H^{(t)}_4+\frac{1}{x_5-f_3}\dot{Z}^{(t)}_{10}\right).
\end{align}
The construction of the coded increment $\mathbf{Q}^{(t)}_n, n\in[N]\setminus\mathcal{D}^{(t)}$ is thus 
\begin{align}
    \mathbf{Q}^{(t)}_n = \left(\mathbf{C}\dot{\mathbf{M}}^{(t)}\right)(n, :).
\end{align}
Similarly, upon receiving the coded increment, Server $n, n\in[N]\setminus\mathcal{D}^{(t)}$ updates its storage according to the following equation
\begin{align}
    \mathbf{S}_n^{(t+1)}=&\mathbf{S}^{(t)}_n+\mathbf{Q}^{(t)}_n.
\end{align}
While for all $n\in\mathcal{D}^{(t)}$, $\mathbf{S}_n^{(t+1)}=\mathbf{S}^{(t)}_n$. Since for this construction, it is also guaranteed that 
\begin{align}\label{eq:examc2H}
    \mathbf{C}(\mathcal{D}^{(t)},:)\dot{\mathbf{M}}^{(t)}=\mathbf{0},
\end{align}
at time slot $t+1$, the storage at Server $n, n\in[N]$ can be written uniformly as follows
\begin{align}
&\mathbf{S}_n^{(t+1)}\notag\\
=&\left(\mathbf{C}\mathbf{M}^{(t)}\right)(n,:)+\left(\mathbf{C}\dot{\mathbf{M}}^{(t)}\right)(n,:)\\
=&\mathbf{C}(n,:)\left(
\left[ \begin{NiceArray}{ccc|c|cc}
W^{(t)}_1 & W^{(t)}_2 & W^{(t)}_3 & \Block[transparent, fill=Plum!20,rounded-corners]{3-1}{} {Z^{(t)}_1}  & \Block[transparent, fill=Cyan!20,rounded-corners]{2-2}{} {Z^{(t)}_4}  & {Z^{(t)}_5} \\
W^{(t)}_4 & W^{(t)}_5 & W^{(t)}_6 & {Z^{(t)}_2}  & {Z^{(t)}_6}  & {Z^{(t)}_8}\\
W^{(t)}_7 & W^{(t)}_8 & W^{(t)}_9 & {Z^{(t)}_3}  & {Z^{(t)}_9}  & {Z^{(t)}_{10}} \\
W^{(t)}_{10} & W^{(t)}_{11} & W^{(t)}_{12} & {Z^{(t)}_7}  & {Z^{(t)}_{11}}  & {Z^{(t)}_{12}}\\
\Block[transparent, fill=Plum!20,rounded-corners]{1-3}{}
{Z^{(t)}_1}  & {Z^{(t)}_2}  & {Z^{(t)}_{3}}  &\Block[transparent, fill=Cyan!20,rounded-corners]{1-1}{} {Z^{(t)}_8}  & 0 & 0\\
\Block[transparent, fill=Cyan!20,rounded-corners]{1-3}{}{Z^{(t)}_4}  & {Z^{(t)}_5}  & {Z^{(t)}_6}  & 0 & 0 & 0
\end{NiceArray}\right]+
\left[ \begin{NiceArray}{ccc|c|cc}
\Delta^{(t)}_1 & \Delta^{(t)}_2 & \Delta^{(t)}_3 &\Block[transparent, fill=Plum!20,rounded-corners]{3-1}{} {\dot{Z}^{(t)}_1} &\Block[transparent, fill=Cyan!20,rounded-corners]{2-2}{} {H^{(t)}_1} & {H^{(t)}_2}\\
\Delta^{(t)}_4 & \Delta^{(t)}_5 & \Delta^{(t)}_6 & {\dot{Z}^{(t)}_2} & {H^{(t)}_3} & {H^{(t)}_4}\\
\Delta^{(t)}_7 & \Delta^{(t)}_8 & \Delta^{(t)}_9 & {\dot{Z}^{(t)}_3} & {\dot{Z}^{(t)}_9}  & {\dot{Z}^{(t)}_{10}}\\
\Delta^{(t)}_{10} & \Delta^{(t)}_{11} & \Delta^{(t)}_{12} & {\dot{Z}^{(t)}_7}  & {H^{(t)}_5}  & {H^{(t)}_6}\\
\Block[transparent, fill=Plum!20,rounded-corners]{1-3}{}{\dot{Z}^{(t)}_1}  & {\dot{Z}^{(t)}_2} & {\dot{Z}^{(t)}_3} &\Block[transparent, fill=Cyan!20,rounded-corners]{1-1}{} {H^{(t)}_4} & 0 & 0\\
\Block[transparent, fill=Cyan!20,rounded-corners]{1-3}{}{H^{(t)}_1} & {H^{(t)}_2} & {H^{(t)}_3}  & 0 & 0 & 0
\end{NiceArray}\right]\right)\\
=&\notag\mathbf{C}(n,:)
\left[ \begin{NiceArray}{ccc|c|cc}
W^{(t)}_1+\Delta^{(t)}_1 & W^{(t)}_2+\Delta^{(t)}_2 & W^{(t)}_3+\Delta^{(t)}_3 &\Block[transparent, fill=Plum!20,rounded-corners]{3-1}{} {Z^{(t)}_1+\dot{Z}^{(t)}_1}  &\Block[transparent, fill=Cyan!20,rounded-corners]{2-2}{} {Z^{(t)}_4+H^{(t)}_1}  & {Z^{(t)}_5+H^{(t)}_2} \\
W^{(t)}_4+\Delta^{(t)}_4 & W^{(t)}_5+\Delta^{(t)}_5 & W^{(t)}_6+\Delta^{(t)}_6 & {Z^{(t)}_2+\dot{Z}^{(t)}_2}  & {Z^{(t)}_6+H_3}  & {Z^{(t)}_8+H^{(t)}_4}\\
W^{(t)}_7+\Delta^{(t)}_7 & W^{(t)}_8+\Delta^{(t)}_8 & W^{(t)}_9+\Delta^{(t)}_9 & {Z^{(t)}_3+\dot{Z}^{(t)}_3}  & {Z^{(t)}_9+\dot{Z}^{(t)}_9}  & {Z^{(t)}_{10}+\dot{Z}^{(t)}_{10}} \\
W^{(t)}_{10}+\Delta^{(t)}_{10} & W^{(t)}_{11}+\Delta^{(t)}_{11} & W^{(t)}_{12}+\Delta^{(t)}_{12} & {Z^{(t)}_7+\dot{Z}^{(t)}_7}  & {Z^{(t)}_{11}+H^{(t)}_5}  & {Z^{(t)}_{12}+H^{(t)}_6}\\
\Block[transparent, fill=Plum!20,rounded-corners]{1-3}{}{Z^{(t)}_1+\dot{Z}^{(t)}_1}  & {Z^{(t)}_2+\dot{Z}^{(t)}_2}  & {Z^{(t)}_{3}+\dot{Z}^{(t)}_3}  &\Block[transparent, fill=Cyan!20,rounded-corners]{1-1}{} {Z^{(t)}_8+H^{(t)}_4}  & 0 & 0\\
\Block[transparent, fill=Cyan!20,rounded-corners]{1-3}{}{Z^{(t)}_4+H^{(t)}_1}  & {Z^{(t)}_5+H^{(t)}_2}  & {Z^{(t)}_6+H^{(t)}_3}  & 0 & 0 & 0
\end{NiceArray}\right]\\
=&\mathbf{C}(n,:)
\left[ \begin{NiceArray}{ccc|c|cc}
W^{(t+1)}_1 & W^{(t+1)}_2 & W^{(t+1)}_3 & \Block[transparent, fill=Plum!20,rounded-corners]{3-1}{} {Z^{(t+1)}_1}  & \Block[transparent, fill=Cyan!20,rounded-corners]{2-2}{} {Z^{(t+1)}_4}  & {Z^{(t+1)}_5} \\
W^{(t+1)}_4 & W^{(t+1)}_5 & W^{(t+1)}_6 & {Z^{(t+1)}_2}  & {Z^{(t+1)}_6}  & {Z^{(t+1)}_8}\\
W^{(t+1)}_7 & W^{(t+1)}_8 & W^{(t+1)}_9 & {Z^{(t+1)}_3}  & {Z^{(t+1)}_9}  & {Z^{(t+1)}_{10}} \\
W^{(t+1)}_{10} & W^{(t+1)}_{11} & W^{(t+1)}_{12} & {Z^{(t+1)}_7}  & {Z^{(t+1)}_{11}}  & {Z^{(t+1)}_{12}}\\
\Block[transparent, fill=Plum!20,rounded-corners]{1-3}{}
{Z^{(t+1)}_1}  & {Z^{(t+1)}_2}  & {Z^{(t+ 1)}_{3}}  &\Block[transparent, fill=Cyan!20,rounded-corners]{1-1}{} {Z^{(t+1)}_8}  & 0 & 0\\
\Block[transparent, fill=Cyan!20,rounded-corners]{1-3}{}{Z^{(t+1)}_4}  & {Z^{(t+1)}_5}  & {Z^{(t+1)}_6}  & 0 & 0 & 0
\end{NiceArray}\right]\\
=&\mathbf{C}(n,:)\mathbf{M}^{(t+1)}.
\end{align}
Now it is obvious that the storage at time slot $t+1$ is in accordance with \eqref{eq:nstor0} and the following recurrence relation holds
\begin{align}
    Z_i^{(t+1)}&=Z_i^{(t)}+\dot{Z}^{(t)}_i, &i&=1,2,3,7,9,10\\
    Z_i^{(t+1)}&=Z_i^{(t)}+H_{i-3}^{(t)}, &i&=4,5,6\\
    Z_8^{(t+1)}&=Z_8^{(t)}+H_{4}^{(t)},\\
    Z_i^{(t+1)}&=Z_i^{(t)}+H_{i-6}^{(t)}, &i&=11,12.
\end{align}
For this example case, the normalized upload cost is 
\begin{align}
    C_u^{(t)}=\frac{5\times 6}{12}=5/2,
\end{align}
which also achieves the converse bound $(N-|\mathcal{D}^{(t)}|)/(R_r-X^{(t)}-|\mathcal{D}^{(t)}|)=5/2$.

\subsection{The General Scheme}
First, let us define a set of notations that is needed in the construction of our achievability scheme as follows.
\begin{align}
\begin{array}{lll}
     G= N-R_r+1,& & \\
     \alpha_i= N-R_r+K_c+1-i,& \beta_i= N+1-i,& \forall i\in[G],\\
     \gamma_1=L/\alpha_1,& \gamma_i= L/(\alpha_{i-1}\alpha_i),& \forall i\in[2:G],\\
     \lambda_0= 0,& \lambda_i= L/\alpha_i,& \forall i\in[G].
\end{array}
\end{align}
Also, let us set the length of message $L={\rm lcm}(\alpha_1,\alpha_2,\cdots,\alpha_{G})$. Note that for all $i\in[G]$, $\gamma_1+\gamma_2+\cdots+\gamma_i=\lambda_{i}$.
Let $x_1,x_2,\cdots,x_N,f_1,f_2,\cdots,f_{N}$ be a total of $2N$ distinct elements from a finite field $\mathbb{F}_q, q\geq 2N$.

Before the general achievability proof, we first present the following algorithm $\FuncSty{SCGen}$ that outputs a matrix $\mathbf{M}$ given the input symbols $\mathbf{W}, \mathbf{Z}_1, \mathbf{Z}_2, \cdots, \mathbf{Z}_G$ as per a staircase structure. The generation of the staircase structure lies at the heart of our achievability scheme.

\begin{algorithm}
\DontPrintSemicolon
\caption{Generation of the Staircase Structure}\label{alg:mat}
\KwIn{$\mathbf{W}\in\mathbb{F}_q^{L},$ $\mathbf{Z}_i\in\mathbb{F}_q^{(R_r-K_c)\times\gamma_i},i\in[G]$}
\KwOut{$\mathbf{M}=\left[\mathbf{M}_1,\mathbf{M}_2,\cdots,\mathbf{M}_G\right]$}

\SetKwFunction{FMain}{SCGen}
    \SetKwProg{Fn}{Function}{:}{}
    \Fn{\FMain{$\mathbf{W}, (\mathbf{Z}_i)_{i\in[G]}$}}{
        \ForEach{$i\in[G]$}{
            \eIf{ $i=1$}{
                $\mathbf{M}_1([\alpha_1],:)\gets\FuncSty{Reshape}(\mathbf{W},\alpha_1,\gamma_1)$ \Comment{Reshape the vector into an $\alpha_1$ by $\gamma_1$ matrix}\;
                $\mathbf{M}_1([\alpha_1+1:N],:)\gets\mathbf{Z}_1$\;
            }{
                $\mathbf{D}_{i-1}\gets\FuncSty{Reshape}(\left[\mathbf{M}_1(R_r+i-1,:),\cdots,\mathbf{M}_j(R_r+i-j,:),\cdots,\mathbf{M}_{i-1}(R_r+1,:)\right],\alpha_i,\gamma_i)$\; 
                $\mathbf{M}_i([\alpha_i],:)\gets\mathbf{D}_{i-1}$\;
                $\mathbf{M}_i([\alpha_i+1:\beta_i],:)\gets\mathbf{Z}_{i}$\;
                $\mathbf{M}_i([\beta_i+1:N],:)\gets\mathbf{0}_{(N-\beta_i)\times \gamma_i}$\;
            }
        } 
        \textbf{return} $\mathbf{M}=\left[\mathbf{M}_1,\mathbf{M}_2,\cdots,\mathbf{M}_G\right]$\;
    }
    \textbf{End Function}
\end{algorithm}

The structure of $\mathbf{M}_1,\mathbf{M}_2,\cdots,\mathbf{M}_G$ in Algorithm \ref{alg:mat} is illustrated in Figure \ref{fig:Mi}, where for each $\mathbf{D}_i, i\in[G-1]$, the elements of $\mathbf{M}_1,\mathbf{M}_2,\cdots,\mathbf{M}_{i}$ that it replicates are shaded in the same color.

\begin{figure}[!h]
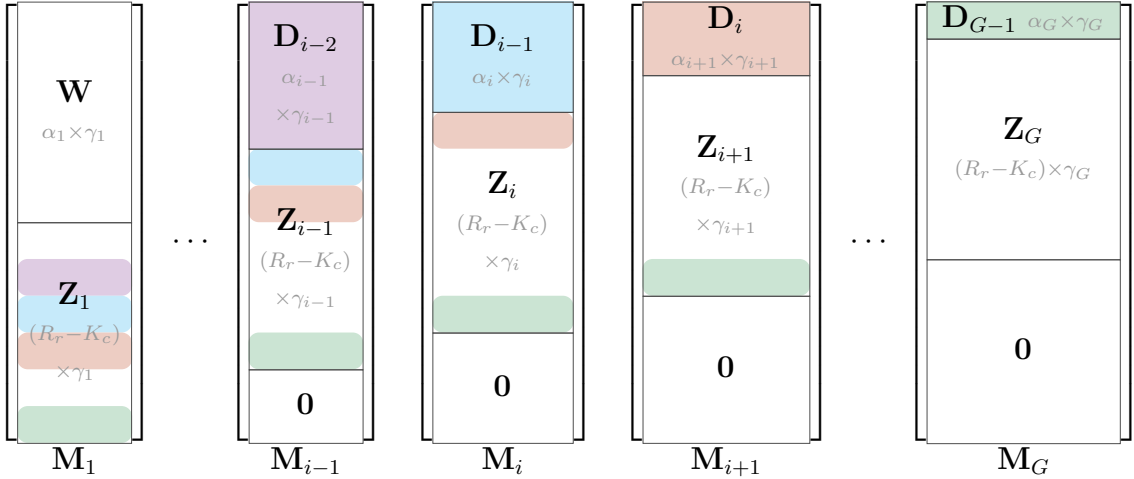

\centering
    \begin{align}
        \begin{array}{c}
        \left[\begin{NiceArray}{cccc}
        \Block[transparent, line-width=0.3pt, draw=black!70]{6-4}{\mathbf{W}\\ \textcolor{black!40}{\scriptstyle\alpha_1\times\gamma_1}}
        &&\textcolor{white}{\mathbf{00}}&\\
        &&&\\
        &&&\\
        &&&\\
        &&&\\
        &&&\\ 
        \Block[transparent, line-width=0.3pt, draw=black!70]{6-4}{\mathbf{Z}_1\\ \textcolor{black!40}{\scriptstyle (R_r-K_c)}\\ \textcolor{black!40}{\scriptstyle\times\gamma_1}}
        &&&\\
        \Block[transparent, fill=Plum!20,rounded-corners]{1-4}{}
        &&&\\
        \Block[transparent, fill=Cyan!20,rounded-corners]{1-4}{}
        &&&\\
        \Block[transparent, fill=Maroon!20,rounded-corners]{1-4}{}
        &&&\\
        &&&\\
        \Block[transparent, fill=ForestGreen!20,rounded-corners]{1-4}{}
        &&&
        \end{NiceArray}\right]\\
        \mathbf{M}_1
        \end{array}\cdots
        \begin{array}{c}
        \left[\begin{NiceArray}{cccc}
        \Block[transparent, fill=Plum!20, draw=black!70]{4-4}{\mathbf{D}_{i-2}\\ \textcolor{black!40}{\scriptstyle\alpha_{i-1}}\\ \textcolor{black!40}{\scriptstyle\times\gamma_{i-1}}}
        &&&\\
        &&&\\
        &&&\\
        &&&\\
        \Block[transparent, line-width=0.3pt, draw=black!70]{6-4}{\mathbf{Z}_{i-1}\\ \textcolor{black!40}{\scriptstyle (R_r-K_c)}\\ \textcolor{black!40}{\scriptstyle\times\gamma_{i-1}}}
        \Block[transparent, fill=Cyan!20,rounded-corners]{1-4}{}
        &&&\\
        \Block[transparent, fill=Maroon!20,rounded-corners]{1-4}{}
        &&&\\
        &&&\\
        &&&\\
        &&&\\
        \Block[transparent, fill=ForestGreen!20,rounded-corners]{1-4}{}
        &&&\\
        \Block[transparent, line-width=0.3pt, draw=black!70]{2-4}{\mathbf{0}}
        &&&\\
        \textcolor{white}{\mathbf{00}}&&&
        \end{NiceArray}\right]\\
        \mathbf{M}_{i-1}
        \end{array}
        \begin{array}{c}
        \left[\begin{NiceArray}{ccccc}
        \Block[transparent, fill=Cyan!20, draw=black!70]{3-5}{\mathbf{D}_{i-1}\\ \textcolor{black!40}{\scriptstyle\alpha_{i}\times\gamma_{i}}}
        &&&&\\
        &&&&\\
        &&&&\\
        \Block[transparent, line-width=0.3pt, draw=black!70]{6-5}{\mathbf{Z}_{i}\\ \textcolor{black!40}{\scriptstyle (R_r-K_c)}\\ \textcolor{black!40}{\scriptstyle\times\gamma_{i}}}
        \Block[transparent, fill=Maroon!20,rounded-corners]{1-5}{}
        &&&&\\
        &&&&\\
        &&&&\\
        &&&&\\
        &&&&\\
        \Block[transparent, fill=ForestGreen!20,rounded-corners]{1-5}{}
        &&&&\\
        \Block[transparent, line-width=0.3pt, draw=black!70]{3-5}{\mathbf{0}}
        &&&&\\
        &&&&\\
        \textcolor{white}{\mathbf{00}}&&&&
        \end{NiceArray}\right]\\
        \mathbf{M}_{i}
        \end{array}
        \begin{array}{c}
        \left[\begin{NiceArray}{cccccc}
        \Block[transparent, fill=Maroon!20, draw=black!70]{2-6}{\mathbf{D}_{i}\\ \textcolor{black!40}{\scriptstyle\alpha_{i+1}\times\gamma_{i+1}}}
        &&&&&\\
        &&&&&\\
        \Block[transparent, line-width=0.3pt, draw=black!70]{6-6}{\mathbf{Z}_{i+1}\\ \textcolor{black!40}{\scriptstyle (R_r-K_c)}\\ \textcolor{black!40}{\scriptstyle\times\gamma_{i+1}}}
        &&&&&\\
        &&&&&\\
        &&&&&\\
        &&&&&\\
        &&&&&\\
        \Block[transparent, fill=ForestGreen!20,rounded-corners]{1-6}{}
        &&&&&\\
        \Block[transparent, line-width=0.3pt, draw=black!70]{4-6}{\mathbf{0}}
        &&&&&\\
        &&&&&\\
        &&&&&\\
        \textcolor{white}{\mathbf{00}}&&&&&
        \end{NiceArray}\right]\\
        \mathbf{M}_{i+1}
        \end{array}\cdots
        \begin{array}{c}
        \left[\begin{NiceArray}{ccccccc}
        \Block[transparent, fill=ForestGreen!20, draw=black!70]{1-7}{\mathbf{D}_{G-1}~ \textcolor{black!40}{\scriptstyle \alpha_{G}\times\gamma_{G}}}
        &&&&&&\\
        \Block[transparent, line-width=0.3pt, draw=black!70]{6-7}{\mathbf{Z}_{G}\\ \textcolor{black!40}{\scriptstyle (R_r-K_c)} \textcolor{black!40}{\scriptstyle\times\gamma_{G}}}
        &&&&&&\\
        &&&&&&\\
        &&&&&&\\
        &&&&&&\\
        &&&&&&\\
        &&&&&&\\
         \Block[transparent, line-width=0.3pt, draw=black!70]{5-7}{\mathbf{0}}
        &&&&&&\\
        &&&&&&\\
        &&&&&&\\
        &&&&&&\\
        \textcolor{white}{\mathbf{00}}&&&&&&
        \end{NiceArray}\right]\\
        \mathbf{M}_{G}
        \end{array}\notag
    \end{align}
    \caption{The structure of $\mathbf{M}_1,\mathbf{M}_2,\cdots,\mathbf{M}_G$ as per Algorithm \ref{alg:mat}. For each $\mathbf{D}_i, i=1,2,\cdots,G-1$, it replicates the elements of $\mathbf{M}_1,\mathbf{M}_2,\cdots,\mathbf{M}_{i}$, shaded in the same color. Note that the block $\mathbf{W}$ illustrated in $\mathbf{M}_1$ is the reshaped version of the message vector $\mathbf{W}$.}
    \label{fig:Mi}
\end{figure}

According to the algorithm and Figure \ref{fig:Mi}, the following propositions are straightforward.
\begin{proposition}\label{prop:copy}
    For all $i\in[G],j\in[R_r+1:\beta_i-1],k\in[j+1:\beta_i]$, there exist a total of $(k-j+1)\gamma_i$ elements in the matrices $(\mathbf{M}_{i'})_{i'\in[i+j-R_r:i+k-R_r]}$ that replicate all elements in the matrix $\mathbf{M}_{i}([j:k],:)$.
\end{proposition}
\begin{proposition}\label{prop:ssim}
    For all $i\in[G-R_r+K_c:G]$, if $\mathbf{M}=\FuncSty{SCGen}(\mathbf{W}, (\mathbf{Z}_{i'})_{i'\in[G]})$, where
    \begin{align}
        &\mathbf{Z}_{i'}=\begin{bmatrix}
            \dot{\mathbf{Z}}_{i'}\\
            \mathbf{0}_{(G-i)\times \gamma_{i'}}
        \end{bmatrix},~{i'}\in[i],\\
        &\mathbf{Z}_{i'}=\mathbf{0}_{(R_r-K_c)\times \gamma_{i'}},~{i'}\in[i+1:G],
    \end{align}
    and for all ${i'}\in[i]$, $\dot{\mathbf{Z}}_{i'}\in\mathbb{F}_q^{(R_r-K_c-G+i)\times\gamma_{i'}}$, then we have $\mathbf{M}(:,[\lambda_{i}+1:\lambda_{G}])=\mathbf{0}_{N\times(\lambda_G-\lambda_i)}$.
\end{proposition}

\begin{proposition}\label{prop:add}
    If $\mathbf{M}=\FuncSty{SCGen}(\mathbf{W}, (\mathbf{Z}_i)_{i\in[G]})$ and $\mathbf{M}'=\FuncSty{SCGen}(\mathbf{W}', (\mathbf{Z}'_i)_{i\in[G]})$, then $\mathbf{M}+\mathbf{M}'=\FuncSty{SCGen}(\mathbf{W}+\mathbf{W}', (\mathbf{Z}_i+\mathbf{Z}'_i)_{i\in[G]})$.
\end{proposition}

\subsubsection{Construction of the Storage}
Recall that the $N$ distributed servers must form an $(K_c, R_r, N)$-secure storage at any time slot $t, t\in\mathbb{N}$. To this end, let the message $\mathbf{W}^{(t)}=[W_1^{(t)},W_2^{(t)},\cdots,W_L^{(t)}]$ consist of $L$ symbols from the finite field $\mathbb{F}_q$. Let us also define the Cauchy matrix generated by $x_1,x_2,\cdots,x_N,f_1,f_2,\cdots,f_{N}$ as follows
\begin{align}\label{eq:ccmat}
\mathbf{C}=
\begin{bmatrix}
\frac{1}{x_1-f_1}&\frac{1}{x_1-f_2}&\cdots&\frac{1}{x_1-f_{N}}\\
\frac{1}{x_2-f_1}&\frac{1}{x_2-f_2}&\cdots&\frac{1}{x_2-f_{N}}\\
\vdots&\vdots&\cdots&\vdots\\
\frac{1}{x_N-f_1}&\frac{1}{x_N-f_2}&\cdots&\frac{1}{x_N-f_{N}}
\end{bmatrix}.
\end{align}
Let $\mathbf{M}^{(t)}=\left[\mathbf{M}^{(t)}_1,\mathbf{M}^{(t)}_2,\cdots,\mathbf{M}^{(t)}_G\right]=\FuncSty{SCGen}(\mathbf{W}^{(t)}, (\mathbf{Z}^{(t)}_i)_{i\in[G]})$, 
where for all $i\in[G]$, $\mathbf{M}^{(t)}_i\in\mathbb{F}_q^{N\times \gamma_i}$, and $\mathbf{Z}^{(t)}_i\in\mathbb{F}_q^{(R_r-K_c)\times\gamma_i}$. Besides, for $t=0$, let $(\mathbf{Z}^{(0)}_i)_{i\in[G]}$ be uniformly i.i.d. over $\mathbb{F}_q^{(R_r-K_c)\times\gamma_1\times\gamma_2\times\cdots\times\gamma_G}$, independent of the (initial) message $\mathbf{W}^{(0)}$. Then for our RDCDS scheme, at any time slot $t\in\mathbb{N}$, the storage at Server $n$, denoted as $\mathbf{S}_n^{(t)}$, is
\begin{align}\label{eq:nstor}
    \mathbf{S}_n^{(t)} = \left(\mathbf{C}\mathbf{M}^{(t)}\right)(n,:),
\end{align}
i.e., the $n^{th}$ row of $\mathbf{C}\mathbf{M}^{(t)}$. Similarly, the storage at time slot $t=0$, i.e., $\mathbf{S}_n^{(0)}, n\in[N]$, is initialized {\it a priori} by, e.g., the global coordinator, who is responsible for the generation of the initial message $\mathbf{W}^{(0)}$, the initial noise symbols $(\mathbf{Z}^{(0)}_i)_{i\in[G]}$ and the initial storage at the $N$ servers according to the construction. Note that each server stores a total of $\lambda_G=L/\alpha_G=L/K_c$ $q$-ary symbols, so the normalized storage cost is $H(\mathbf{S}_n^{(t)})/L=1/K_c$, which meets the storage cost constraint \eqref{eq:minstor}.

\subsubsection{Execution of the Read Operation}
Assume that the user wishes to execute the read operation at time slot $t$.  
Let $\overline{\mathcal{D}}^{(t)}=[N]\setminus\mathcal{D}^{(t)}$ and $J^{(t)}=N+1-|\overline{\mathcal{D}}^{(t)}|$. Recall that according to the definition of $\beta_i$'s, we have $\beta_{J^{(t)}}=|\overline{\mathcal{D}}^{(t)}|$.
To recover the message $\mathbf{W}^{(t)}$, the user downloads the following symbols from available servers.
\begin{align}
    \mathbf{A}_n^{(t)}=\mathbf{S}_n^{(t)}([\lambda_{J^{(t)}}]), n\in\overline{\mathcal{D}}^{(t)}.
\end{align}

At this point, the user is able to recover the message from the downloads via a successive interference cancellation decoding strategy, as depicted below. Note that according to the construction of the storage, the downloaded symbols can be written in the following matrix form
\begin{align}
\begin{bmatrix}
\mathbf{A}^{(t)}_{\overline{\mathcal{D}}^{(t)}(1)}\\
\mathbf{A}^{(t)}_{\overline{\mathcal{D}}^{(t)}(2)}\\
\vdots\\
\mathbf{A}^{(t)}_{\overline{\mathcal{D}}^{(t)}(|\overline{\mathcal{D}}^{(t)}|)}
\end{bmatrix}
=&\left(\mathbf{C}\mathbf{M}^{(t)}\right)(\overline{\mathcal{D}}^{(t)},[\lambda_{J^{(t)}}])\\
=&\mathbf{C}(\overline{\mathcal{D}}^{(t)},:)\mathbf{M}^{(t)}(:,[\lambda_{J^{(t)}}])\\
=&\mathbf{C}(\overline{\mathcal{D}}^{(t)},:)\left[\mathbf{M}^{(t)}_1,\mathbf{M}^{(t)}_2,\cdots,\mathbf{M}^{(t)}_{J^{(t)}}\right]
\end{align}

Since for all $i\in[J^{(t)}]$, $\mathbf{M}^{(t)}_i([\beta_i+1:N],:)$ are zeros according to Algorithm \ref{alg:mat}, for all $i\in[J^{(t)}]$, we have
\begin{align}
\begin{bmatrix}
\mathbf{A}^{(t)}_{\overline{\mathcal{D}}^{(t)}(1)}([\lambda_{i-1}+1:\lambda_{i}])\\
\mathbf{A}^{(t)}_{\overline{\mathcal{D}}^{(t)}(2)}([\lambda_{i-1}+1:\lambda_{i}])\\
\vdots\\
\mathbf{A}^{(t)}_{\overline{\mathcal{D}}^{(t)}(|\overline{\mathcal{D}}^{(t)}|)}([\lambda_{i-1}+1:\lambda_{i}])
\end{bmatrix}=
    \mathbf{C}(\overline{\mathcal{D}}^{(t)},[\beta_{i}])\mathbf{M}^{(t)}_i([\beta_{i}],:)
\end{align}

Because the constants $x_1,x_2,\cdots,x_N,f_1,f_2,\cdots,f_{N}$ are distinct, all square submatrices of \eqref{eq:ccmat} are of full rank. Thus $\mathbf{M}^{(t)}_{J^{(t)}}([\beta_{J^{(t)}}],:)$ is resolvable by inverting the Cauchy matrix on the RHS of \eqref{eq:decMJ}.
\begin{align}\label{eq:decMJ}
\begin{bmatrix}
\mathbf{A}^{(t)}_{\overline{\mathcal{D}}^{(t)}(1)}([\lambda_{J^{(t)}-1}+1:\lambda_{J^{(t)}}])\\
\mathbf{A}^{(t)}_{\overline{\mathcal{D}}^{(t)}(2)}([\lambda_{J^{(t)}-1}+1:\lambda_{J^{(t)}}])\\
\vdots\\
\mathbf{A}^{(t)}_{\overline{\mathcal{D}}^{(t)}(|\overline{\mathcal{D}}^{(t)}|)}([\lambda_{J^{(t)}-1}+1:\lambda_{J^{(t)}}])
\end{bmatrix}=
    \mathbf{C}(\overline{\mathcal{D}}^{(t)},[\beta_{J^{(t)}}])\mathbf{M}^{(t)}_{J^{(t)}}([\beta_{J^{(t)}}],:).
\end{align}

Then, according to Proposition \ref{prop:copy}, for all $i\in\{J^{(t)}-1,J^{(t)}-2,\cdots,1\}$, if $(\mathbf{M}^{(t)}_{i'})_{i'\in[i+1:J^{(t)}]}$ are decoded, then the elements of $\mathbf{M}^{(t)}_i([R_r+1:R_r+J^{(t)}-i],:)$ are also recovered. Recall that $|[\beta_i]\setminus[R_r+1:R_r+J^{(t)}-i]|=\beta_i-(J^{(t)}-i)=|\overline{\mathcal{D}}^{(t)}|$, therefore $\mathbf{M}^{(t)}_i([\beta_i]\setminus[R_r+1:R_r+J^{(t)}-i],:)$ can be decoded by subtracting the recovered symbols and inverting the square Cauchy matrix on the RHS of \eqref{eq:decMi} as follows.
\begin{align}\label{eq:decMi}
&\begin{bmatrix}
\mathbf{A}^{(t)}_{\overline{\mathcal{D}}^{(t)}(1)}([\lambda_{i-1}+1:\lambda_{i}])\notag\\
\mathbf{A}^{(t)}_{\overline{\mathcal{D}}^{(t)}(2)}([\lambda_{i-1}+1:\lambda_{i}])\\
\vdots\\
\mathbf{A}^{(t)}_{\overline{\mathcal{D}}^{(t)}(|\overline{\mathcal{D}}^{(t)}|)}([\lambda_{i-1}+1:\lambda_{i}])
\end{bmatrix}\notag\\
&-\mathbf{C}(\overline{\mathcal{D}}^{(t)},[R_r+1:R_r+J^{(t)}-i])\mathbf{M}^{(t)}_i([R_r+1:R_r+J^{(t)}-i],:)\notag\\
    =&\mathbf{C}(\overline{\mathcal{D}}^{(t)},[\beta_i]\setminus[R_r+1:R_r+J^{(t)}-i])
    \mathbf{M}^{(t)}_i([\beta_i]\setminus[R_r+1:R_r+J^{(t)}-i],:).
\end{align}
In other words, $\mathbf{M}_i^{(t)}$ is recovered at this point. Now it is evident to see that $\mathbf{M}_{J^{(t)}}^{(t)}, \mathbf{M}_{J^{(t)}-1}^{(t)},\cdots, \mathbf{M}_{1}^{(t)}$ are recoverable by induction on $i$, from which the message $\mathbf{W}^{(t)}$ is decodable. Finally, since a total of $L$ symbols of the desired message are recovered from a total of $|\overline{\mathcal{D}}^{(t)}|\times \lambda_{J^{(t)}}$ downloaded symbols, the normalized download cost is 
\begin{align}
    C^{(t)}_r=\frac{|\overline{\mathcal{D}}^{(t)}|\lambda_{J^{(t)}}}{L}
    =\frac{N-|\mathcal{D}^{(t)}|}{N-R_r+K_c-|\mathcal{D}^{(t)}|},
\end{align}
which matches the converse bound.

\subsubsection{Execution of the Update Operation}
Let us assume that the user wishes to update the message to the distributed servers with the generated increment $\mathbf{\Delta}^{(t)}=[\Delta^{(t)}_1,\Delta^{(t)}_2,\cdots,\Delta^{(t)}_L]$ at time slot $t$. To this end, let us define $G^{(t)}=N-2R_r+K_c+X^{(t)}+|\mathcal{D}^{(t)}|+1$, and thus according to the definition of $\alpha_i$'s, we have $\alpha_{G^{(t)}}+X^{(t)}+|\mathcal{D}^{(t)}|=R_r$. Let $\dot{\mathbf{M}}^{(t)}=\left[\dot{\mathbf{M}}^{(t)}_1,\cdots,\dot{\mathbf{M}}^{(t)}_{G}\right]=\FuncSty{SCGen}(\bm{\Delta}^{(t)}, (\dot{\mathbf{Z}}^{(t)}_i)_{i\in[G]})$, where for all $i\in[G]$, $\dot{\mathbf{M}}^{(t)}_i\in\mathbb{F}_q^{N\times\gamma_i}$. Besides, $\dot{\mathbf{Z}}^{(t)}_1, \dot{\mathbf{Z}}^{(t)}_2, \cdots, \dot{\mathbf{Z}}^{(t)}_G$ are defined as follows
\begin{align}
\dot{\mathbf{Z}}^{(t)}_i&=
    \begin{bmatrix}
        \ddot{\mathbf{Z}}^{(t)}_{i}\\
        \mathbf{H}^{(t)}_i\\
        \mathbf{0}_{(R_r-K_c-X^{(t)}-|\mathcal{D}^{(t)}|)\times\gamma_i}
    \end{bmatrix},~\forall i\in[G^{(t)}]\label{eq:dzt1}\\
    \dot{\mathbf{Z}}^{(t)}_i&=\mathbf{0}_{(R_r-K_c)\times\gamma_i},~\forall i\in[G^{(t)}+1:G],\label{eq:dzt2}
\end{align}
where $(\ddot{\mathbf{Z}}^{(t)}_i)_{i\in[G^{(t)}]}$ are uniformly i.i.d. over $\mathbb{F}_q^{X^{(t)}\times\gamma_1\times\gamma_2\times\cdots\times \gamma_{G^{(t)}}}$ and independent of the increment $\boldsymbol{\Delta}^{(t)}$, and $\mathbf{H}^{(t)}_{1}, \mathbf{H}^{(t)}_{2}, \cdots, \mathbf{H}^{(t)}_{G^{(t)}}$ are to be constructed such that $\left(\mathbf{C}\dot{\mathbf{M}}^{(t)}\right)(n,:)=\mathbf{0}_{1\times \lambda_G}$ for all $n\in\mathcal{D}^{(t)}$. The proof of the existence of such $\mathbf{H}^{(t)}_{1}, \mathbf{H}^{(t)}_{2}, \cdots, \mathbf{H}^{(t)}_{G^{(t)}}$ is deferred to the end of this subsection. Now we are ready to construct the coded increment $\mathbf{Q}^{(t)}_n, n\in[N]\setminus\mathcal{D}^{(t)}$ as follows
\begin{align}
    \mathbf{Q}^{(t)}_n &= \left(\mathbf{C}\dot{\mathbf{M}}^{(t)}\right)(n,:),
\end{align}
i.e.,  the $n^{th}$ row of $\mathbf{C}\dot{\mathbf{M}}^{(t)}$. Note that according to \eqref{eq:dzt1}, \eqref{eq:dzt2} and Proposition \ref{prop:ssim}, for all $n\in[N]$, $(\mathbf{C}\dot{\mathbf{M}}^{(t)})(n,[\lambda_{G^{(t)}}+1:\lambda_G])=\mathbf{0}_{1\times (\lambda_G-\lambda_{G^{(t)}})}$. Besides, since $|\mathcal{D}^{(t)}|$ is regarded as a constant during the time slot $t$, it suffices to upload the first $\lambda_{G^{(t)}}$ elements of $\mathbf{Q}^{(t)}_n$ to each of the available servers, and the normalized upload cost can be calculated as follows,
\begin{align}
C^{(t)}_u=\frac{(N-|\mathcal{D}^{(t)}|) \lambda_{G^{(t)}}}{L}=\frac{N-|\mathcal{D}^{(t)}|}{R_r-X^{(t)}-|\mathcal{D}^{(t)}|},
\end{align}
which matches the converse bound. $X^{(t)}$-security follows from the fact that the information symbols are protected by the $\text{MDS}(N,X^{(t)})$-coded noise symbols.

Upon receiving the coded increment, Server $n, n\in[N]\setminus\mathcal{D}^{(t)}$ updates its storage according to the following equation
\begin{align}
    \mathbf{S}_n^{(t+1)}=&\mathbf{S}^{(t)}_n+\mathbf{Q}^{(t)}_n.
\end{align}
Now let us complete the correctness proof by showing that the updated storage $\mathbf{S}_n^{(t+1)}, n\in[N]$ is of the same form as in \eqref{eq:nstor}. Recall that for all $n\in\mathcal{D}^{(t)}$, $\mathbf{S}_n^{(t+1)}=\mathbf{S}^{(t)}_n$ by definition \eqref{eq:uddropstor}, and according to the definition of $\mathbf{H}_i^{(t)},\ i\in[G^{(t)}]$, it is guaranteed that
\begin{align}
    \left(\mathbf{C}\dot{\mathbf{M}}^{(t)}\right)(\mathcal{D}^{(t)},:)=\mathbf{0},
\end{align}
Therefore, the updated storage at Server $n, n\in[N]$ can be written uniformly as follows
\begin{align}
\mathbf{S}^{(t+1)}_n =&\left(\mathbf{C}\left(\mathbf{M}^{(t)}+\dot{\mathbf{M}}^{(t)}\right)\right)(n,:).
\end{align}
The correctness is thus obvious according to Proposition \ref{prop:add}, i.e., $\mathbf{M}^{(t+1)}=\mathbf{M}^{(t)}+\dot{\mathbf{M}}^{(t)}=\FuncSty{SCGen}(\mathbf{W}^{(t+1)}, (\mathbf{Z}^{(t+1)}_i)_{i\in[G]})$, where the following recursion relation holds
\begin{align}
    \mathbf{W}^{(t+1)}=&\mathbf{W}^{(t)}+\bm{\Delta}^{(t)},\\
    \mathbf{Z}^{(t+1)}_i=&\mathbf{Z}^{(t)}_i+\dot{\mathbf{Z}}^{(t)}_i,~ i\in[G],\\
    \mathbf{Z}^{(t+1)}_i=&\mathbf{Z}^{(t)}_i,~ i\in[G^{(t)}+1:G].
\end{align}

Finally, let us show the existence of $\mathbf{H}^{(t)}_{1}, \mathbf{H}^{(t)}_{2}, \cdots, \mathbf{H}^{(t)}_{G^{(t)}}$. According to the construction, it suffices to show the existence of $\mathbf{H}^{(t)}_{1}, \mathbf{H}^{(t)}_{2}, \cdots, \mathbf{H}^{(t)}_{G^{(t)}}$ such that $\left(\mathbf{C}\dot{\mathbf{M}}_{i}^{(t)}\right)(\mathcal{D}^{(t)},:)=\mathbf{0}$ for all $i\in[G^{(t)}]$. First, let us construct $\mathbf{H}^{(t)}_{1}$ as follows.
\begin{align}
    \mathbf{H}^{(t)}_{1}=-&\mathbf{C}(\mathcal{D}^{(t)},[\alpha_1+X^{(t)}+1:\alpha_1+X^{(t)}+|\mathcal{D}^{(t)}|])^{-1}\notag\\
    &\times\left(\mathbf{C}(\mathcal{D}^{(t)},[\alpha_1])\dot{\mathbf{M}}^{(t)}_1([\alpha_1],:)
    +\mathbf{C}(\mathcal{D}^{(t)},[\alpha_1+1:\alpha_1+X^{(t)}])\ddot{\mathbf{Z}}^{(t)}_{1}\right),
\end{align}
where according to Algorithm \ref{alg:mat}, instead of a circular definition, $\dot{\mathbf{M}}^{(t)}_1([\alpha_1],:)$ indeed represents the reshaped version of the increment vector $\boldsymbol{\Delta}^{(t)}$.

It is evident to see that $\left(\mathbf{C}\dot{\mathbf{M}}_1^{(t)}\right)(\mathcal{D}^{(t)},:)=\mathbf{0}$. Now if $\mathbf{H}^{(t)}_{1},\mathbf{H}^{(t)}_{2},\cdots,\mathbf{H}^{(t)}_{i-1}$ are constructed such that for all $i'=1,2,\cdots,i-1$,
\begin{align}
    \left(\mathbf{C}\dot{\mathbf{M}}_{i'}^{(t)}\right)(\mathcal{D}^{(t)},:)=\mathbf{0},
\end{align}
then let us construct $\mathbf{H}^{(t)}_{i}$ as follows.
\begin{align}
    \mathbf{H}^{(t)}_{i}=-&\mathbf{C}(\mathcal{D}^{(t)},[\alpha_i+X^{(t)}+1:\alpha_i+X^{(t)}+|\mathcal{D}^{(t)}|])^{-1}\notag\\
    &\times\left(\mathbf{C}(\mathcal{D}^{(t)},[\alpha_i])\dot{\mathbf{M}}^{(t)}_i([\alpha_i],:)
    +\mathbf{C}(\mathcal{D}^{(t)},[\alpha_i+1:\alpha_i+X^{(t)}])\ddot{\mathbf{Z}}^{(t)}_{i}\right),
\end{align}
where similarly, according to Algorithm \ref{alg:mat}, $\dot{\mathbf{M}}^{(t)}_i([\alpha_i],:)$ represents the reshaped version of the vector $\left[\dot{\mathbf{M}}^{(t)}_1(R_r+i-1,:),\dot{\mathbf{M}}^{(t)}_2(R_r+i-2,:),\cdots,\dot{\mathbf{M}}^{(t)}_{i-1}(R_r+1,:)\right]$, which is given as $\dot{\mathbf{M}}_1^{(t)},\dot{\mathbf{M}}_2^{(t)},\cdots,\dot{\mathbf{M}}_{i-1}^{(t)}$ are now fixed.
According to our construction, we can see that
\begin{align}
    &\left(\mathbf{C}\dot{\mathbf{M}}_i^{(t)}\right)(\mathcal{D}^{(t)},:)\notag\\
    =&\mathbf{C}(\mathcal{D}^{(t)},[\beta_i])\dot{\mathbf{M}}_i^{(t)}([\beta_i],:)\\
    =&\mathbf{C}(\mathcal{D}^{(t)},[\alpha_i])\dot{\mathbf{M}}^{(t)}_i([\alpha_i],:)
    +\mathbf{C}(\mathcal{D}^{(t)},[\alpha_i+1:\alpha_i+X^{(t)}])\ddot{\mathbf{Z}}^{(t)}_{i}\notag\\
    &+\mathbf{C}(\mathcal{D}^{(t)},[\alpha_i+X^{(t)}+1:\alpha_i+X^{(t)}+|\mathcal{D}^{(t)}|])\mathbf{H}^{(t)}_{i}\\
    =&\mathbf{0}.
\end{align}
In other words, the constructed $\mathbf{H}^{(t)}_{1},\mathbf{H}^{(t)}_{2},\cdots,\mathbf{H}^{(t)}_{i}$ now guarantee that for all $i'=1,2,\cdots, i$, $\left(\mathbf{C}\dot{\mathbf{M}}_{i'}^{(t)}\right)(\mathcal{D}^{(t)},:)=\mathbf{0}$. The existence of $\mathbf{H}^{(t)}_{1}, \mathbf{H}^{(t)}_{2}, \cdots, \mathbf{H}^{(t)}_{G^{(t)}}$ is thus concluded by induction on $i$.

\section{Conclusion}\label{sec:conclu}
In this paper, we investigated the problem of robust dynamic coded distributed storage (RDCDS) and established the fundamental limits on the minimum number of available servers required for feasible updates, the minimum download cost for reads, and the minimum upload cost for updates. The results are particularly relevant given the growing demand for robust and efficient read and update functionalities in emerging coded distributed storage applications. An immediate objective for future work is to determine the fundamental limits of {\it private} read and update operations. Additionally, exploring the RDCDS problem under heterogeneous storage and recovery set constraints presents a promising research direction.

\bibliography{ref}
\bibliographystyle{IEEEtran}
\end{document}